\documentclass[journal]{IEEEtran}

\IEEEoverridecommandlockouts

\usepackage[utf8]{inputenc} 
\usepackage[T1]{fontenc}    
\usepackage{url}            
\usepackage{booktabs}       
\usepackage{amsfonts}       
\usepackage{nicefrac}       
\usepackage{microtype}      
\usepackage{xcolor}         
\usepackage[cmex10]{amsmath}

\usepackage{graphicx}   
\usepackage{caption} 
\usepackage{subcaption} 
\usepackage{makecell} 
\usepackage{framed}  
\usepackage{amsmath,amsfonts,amssymb,amsthm,epsfig,epstopdf,url,array}
\usepackage{cases}
\usepackage{bm}     
\usepackage{bbm}
\usepackage{cite}
\usepackage{algorithm} 
\usepackage{algorithmic}
\usepackage{color}
\usepackage{xparse}
\usepackage{mathtools}     
\usepackage{wasysym}
\usepackage[flushleft]{threeparttable}
\newtheorem{theorem}{Theorem}
\newtheorem{lemma}{Lemma}

\newtheorem{assumption}{Assumption}
\newtheorem{remark}{Remark}
\def\BibTeX{{\rm B\kern-.05em{\sc i\kern-.025em b}\kern-.08em
    T\kern-.1667em\lower.7ex\hbox{E}\kern-.125emX}}

\usepackage{amsmath,amsfonts,bm}

\def\eqref#1{equation~\ref{#1}}

\def\1{\bm{1}}

\DeclareMathAlphabet{\mathsfit}{\encodingdefault}{\sfdefault}{m}{sl}
\SetMathAlphabet{\mathsfit}{bold}{\encodingdefault}{\sfdefault}{bx}{n}

\newcommand{\equa}[1]{
\begin{equation}
\begin{aligned}
	#1
\end{aligned}	
\end{equation}
}

\begin{document}

    \makeatletter
    \newcommand{\linebreakand}{%
      \end{@IEEEauthorhalign}
      \hfill\mbox{}\par
      \mbox{}\hfill\begin{@IEEEauthorhalign}
    }
    \makeatother

 \title{Orchestrating Federated Learning in Space-Air-Ground Integrated Networks: Adaptive Data Offloading and Seamless Handover}

\author{
Dong-Jun~Han,~\IEEEmembership{Member,~IEEE,}
Wenzhi Fang, 
Seyyedali Hosseinalipour,~\IEEEmembership{Member,~IEEE,}\\
Mung~Chiang,~\IEEEmembership{Fellow,~IEEE,}  
Christopher G. Brinton,~\IEEEmembership{Senior Member,~IEEE}   
\thanks{This work was supported in part by
the Defense Advanced Research Projects Agency (DARPA) under Grant
D22AP00168, in part by the National Science Foundation (NSF) under Grant
CNS-2212565, and in part by the Office of Naval Research (ONR) under
Grant N000142112472.

D.-J. Han is with the Department of Computer Science and Engineering, Yonsei University, Seoul, South Korea (email: djh@yonsei.ac.kr).

W. Fang, M. Chiang and C. G. Brinton are with the Elmore Family School of Electrical and Computer Engineering, 
Purdue University, West Lafayette, USA (e-mail: \{fang375, 
chiang,  cgb\}@purdue.edu).

S. Hosseinalipour is with the Department of Electrical Engineering,
University at Buffalo--SUNY, NY, USA (email: alipour@buffalo.edu).
}
} 
\maketitle

 \bstctlcite{IEEEexample:BSTcontrol}

\begin{abstract}  
Devices located in remote regions often lack  coverage from  well-developed terrestrial communication infrastructure.    This not only prevents them from experiencing high quality communication services but also hinders the delivery of machine learning   services in  remote regions. In this paper, we propose a new federated learning (FL) methodology tailored to  space-air-ground integrated networks (SAGINs) to tackle this issue.  Our approach strategically leverages the nodes within space and air layers as both (i) edge computing units and (ii)  model aggregators during the FL process, addressing the challenges that arise from the limited computation powers of ground devices and the absence of terrestrial base stations in the target region. The key idea behind our  methodology is the  adaptive data offloading and handover procedures  that incorporate various network dynamics  in  
 SAGINs, including the mobility,    heterogeneous computation powers, and inconsistent coverage times of incoming satellites.  We  analyze the  latency    of our scheme  and  develop an  adaptive data offloading optimizer, and also  characterize the theoretical convergence  bound of our  proposed algorithm.  Experimental results confirm the advantage of  our SAGIN-assisted FL methodology  in terms of training time  and test accuracy compared with  various baselines. 
\end{abstract} 

\section{Introduction}\label{sec:intro}

As the proliferation of edge devices, including mobile phones, smart vehicles, and Internet-of-Things (IoT) sensors, continues to escalate, they generate vast quantities of data at the wireless edge. In response to this surge, federated learning (FL) \cite{mcmahan2017communication, kairouz2021advances, li2020federated_survey} has emerged as a powerful method for harnessing these distributed data sources to train machine learning (ML) models. Over recent years, FL has garnered significant attention and has been rigorously explored across various configurations: from single-server environments~\cite{mcmahan2017communication, wang2019adaptive}, hierarchical structures~\cite{liu2020client, abad2020hierarchical}, to decentralized networks~\cite{roy2019braintorrent, lalitha2018fully}. This body of research, spanning foundational studies to  
implementations, underscores the adaptability and 
 potential of FL in optimizing data-driven insights at the network edge.
 
\vspace{-0.5mm}

 \subsection{Motivation and Key Questions}

  Despite the advances in FL frameworks, they mostly rely heavily on terrestrial communication infrastructures for model aggregation  
   during the training process. This reliance renders most existing FL methods unsuitable for areas lacking terrestrial communication facilities. Specifically, many remote regions of the Earth, such as mountains, forests, deserts, and coastal areas, do not have well-developed base stations, even though they are home to numerous ground devices, such as IoT sensors, that collect valuable data. The data gathered in these locations are essential for the development of intelligent services tailored to a variety of applications: 
   (i) Disaster predictions in coast, mountain, and forest areas that lack a base station. To achieve this, FL over data samples collected from various types of sensor devices in these remote regions is required. (ii) Autonomous vehicle applications in rural regions. Since these regions have different traffic patterns compared to urban areas, FL needs to be conducted over data samples of vehicles in rural regions. (iii) Medical applications, which is one of the key use cases of FL.  Hospitals located in different areas of the world may want to collaboratively train a global model for disease prediction. In such cases, hospitals located in rural regions can take advantage of satellites based on our approach. (iv) Smart agriculture across farms where a well-developed terrestrial base station is unavailable. In this use case, FL should be conducted using data samples collected from different farms.
Decentralized FL methods~\cite{roy2019braintorrent, lalitha2018fully}, although designed to mitigate some of these challenges, encounter significant obstacles in environments where communication links between devices are unreliable or non-existent, as often found in disaster-affected or maritime regions. Consequently, there is a need for an FL methodology that is  specifically   
 tailored to  
 remote areas, ensuring that the distributed data collected in those regions can be leveraged to develop intelligent services.

  Space-air-ground integrated networks (SAGINs) have recently emerged as a groundbreaking solution within the wireless communications community~\cite{liu2018space, ye2020space}, aimed at extending wireless coverage across the globe, particularly in isolated and remote regions. In addition to the terrestrial nodes located at the \textit{ground layer},  
    SAGINs leverage satellites in the \textit{space layer} and air nodes, such as unmanned aerial vehicles (UAVs), in the \textit{air layer}. This multi-layered architecture enables SAGINs to either complement or entirely supplant traditional terrestrial networks in delivering communication services.
Furthermore, SAGINs are not limited to providing mere connectivity; they also have the potential to act as edge computing platforms~\cite{shang2021computing, yu2021ec, liu2022energy}. In particular, they can undertake computation tasks offloaded from terrestrial, resource-constrained devices, such as IoT sensors.  
 The integration of SAGINs thus promises not only to bridge the connectivity gap in underserved areas but also to enhance the computational capabilities at the network edge, opening new avenues for advanced applications and services.

  Inspired by the capabilities of SAGINs, this paper sets out to explore the orchestration of FL within SAGINs to facilitate FL in remote areas. This  brings forth a set of novel challenges that are absent in conventional FL implementations over terrestrial networks. Our investigation is driven by  
   research questions aimed at unlocking the full potential of FL in the context of SAGINs. 
 First, how should we optimally utilize the unique components of SAGINs, including satellites, air nodes, and ground devices, during the FL process?  Second, how should we address the network dynamics in SAGINs (e.g., dealing with the mobility, varying computation capacities, and the inconsistent coverage times provided by satellites) 
 during FL? Third, 
can we theoretically guarantee the  convergence of   FL  despite the inherent challenges of SAGINs, such as variable network conditions and limited connectivity? Despite the  importance of deploying FL in remote regions 
 for intelligent service development, 
these questions have been largely overlooked in existing research.   
 Our goal is thus to fill this gap by providing insights and solutions that enable effective FL over SAGINs.

 \subsection{Main Contributions}
 
In this paper, we propose a   FL methodology that   takes advantage of  both computation and communication resources of space/air/terrestrial nodes in SAGINs to provide intelligent ML services over remote areas. Compared to prior 
FL methods that rely on base stations,     
 our approach strategically leverages the space and air nodes as both (i) edge computing units and (ii) ML model aggregators during the FL process,  to address the challenges arising from the limited computation powers of ground devices and the absence of terrestrial base stations in the target remote region.  Under this   framework, we propose an adaptive approach to \textit{optimize data offloading}  depending on the network dynamics of SAGINs, including the inconsistent computation capabilities and coverage times of   low-earth orbit (LEO)  satellites.  
 Considering the  mobility of LEO satellites, we also propose an \textit{optimized  data/model handover strategy}  where each satellite transmits the trained model and its dataset to the next incoming satellite to ensure a seamless ML model training process. By  incorporating the handover delay into our latency modeling, we  optimize the amount of data being offloaded across the   layers in SAGINs during the FL process.

Overall, our  contributions can be summarized as follows:  
 \begin{itemize}
\item \textbf{New methodology:}   We introduce a new SAGIN-based FL methodology  with adaptive data offloading and handover,  which         facilitates  intelligent ML services  in remote areas  without  the need for terrestrial communication infrastructures.  Our scheme strategically utilizes the space and air nodes as edge computing units and model aggregators, and captures  the key features of SAGINs including mobility of satellites,   time-varying resources and coverage times of incoming satellites, hierarchical architecture, and computation resources of space/air/terrestrial nodes.

\item \textbf{Analysis and optimization:}  We analyze the latency of the proposed algorithm, and propose an optimized   inter-layer data offloading scheme and an intra-layer data handover strategy for the space layer  to minimize the delay. 
This optimization process takes into account the data transmission delay, data processing delay, and model aggregation delay altogether, as well as various network dynamics in SAGINs.  
We also analytically characterize the convergence bound of our algorithm, and show that the model converges to a stationary point for   non-convex loss functions even when adaptive data offloading is applied.

\item \textbf{Simulations under practical modeling:}  We provide extensive experiments using three FL benchmark datasets. To simulate real-world scenarios, we adopt the Walker-Star function to model a satellite constellation and measure the coverage time  of each satellite over the target region. Experimental results   demonstrate  that the proposed methodology can achieve the target accuracy much faster with less training latency compared to various  baselines.     
 
\end{itemize}
To the best of our knowledge,  this is one of the earliest works to successfully integrate FL with adaptive  data offloading/handover optimization across  space-air-ground layers, while accounting for various network dynamics specific to SAGINs.  

\subsection{Related Works}
 \textbf{FL over terrestrial networks:}  FL has been actively studied in   terrestrial networks where the server (e.g., base station) aggregates   the client models in the system. Most of them consider a   single-server setup   
\cite{mcmahan2017communication, wang2019adaptive, yang2019scheduling, amiri2020federated, chen2020joint, 
chen2020convergence} while  some researchers also study multi-server scenarios \cite{liu2020client, abad2020hierarchical, lim2021decentralized, lim2021dynamic, han2021fedmes}.  In  \cite{wang2019matcha, roy2019braintorrent, lalitha2018fully,
koloskova2020unified}, the authors   investigate decentralized FL where each client aggregates the models via device-to-device communications with its adjacent clients, without relying on the server. 
Data offloading strategies are also studied in FL where each client offloads a portion of its local dataset to the server \cite{huang2022wireless, ganguly2023multi, hosseinalipour2023parallel}. However, prior works on FL    are mostly inapplicable in remote regions, where well-developed base stations are not available and   communication links between devices are unstable (e.g., disaster or maritime regions). 
Compared to these works, we   facilitate  FL in remote areas by strategically leveraging  non-terrestrial network elements, specifically SAGINs.

\textbf{FL with UAVs or satellites:} Another line of research has explored FL over UAVs \cite{wang2020learning,zhang2020federated,zeng2020federated} or satellites \cite{matthiesen2023federated, so2022fedspace, Razmi1, Razmi2, Razmi3, Elmahallawy1, zhai2023fedleo, Elmahallawy3}, where either UAVs or satellites collect their own datasets and are considered as  clients.  After the local training procedure at UAVs or satellites, model aggregation and synchronization are conducted relying on the ground base station \cite{wang2020learning, zhang2020federated, so2022fedspace, Razmi1, Razmi2, Razmi3, Elmahallawy3}  or directly at the UAVs/satellites \cite{Elmahallawy1, zhai2023fedleo}. The problem setup of these studies differs from ours as we focus  on FL over data samples collected at ground devices located in remote regions. This necessitates interaction between ground devices and nodes in the space/air layers not only for model aggregation (to address the lack of base stations in remote regions) but also for computation offloading (to tackle the limited computation capabilities of ground devices).

 Some previous works \cite{rodrigues2023hybrid,  chen2022satellite, wang2022federated, fang2023olive}  have focused on the  setting  where  ground devices collect data and conduct FL  assisted by the UAVs/satellites, similar to   our problem setup.         Specifically in \cite{wang2022federated},  the  satellite  aggregates the models of ground devices via over-the-air aggregation, without requiring any base stations.  The authors of      \cite{rodrigues2023hybrid}  focus on solving the maze problem using the deep Q network assisted by the satellites. In \cite{chen2022satellite, han2024cooperative}, data offloading has been studied for satellite-assisted FL by considering only the space layer. While these works do not consider SAGINs,  a recent work  \cite{fang2023olive} specifically studied FL considering space-air-ground layers.  
However, the nodes in space and air layers are only used as model aggregators, not as edge computing units. Compared to all prior works, the   contribution of this work is to adaptively optimize data offloading     across different layers and handover within the space layer, while taking into account the  network dynamics specific to SAGINs (e.g.,  heterogeneous coverage time and resource availability of current/incoming satellites)  during FL.

\begin{figure*}[t]
  \centering
 \begin{subfigure}[b]{0.45\textwidth}
         \centering
         \includegraphics[width=\textwidth]{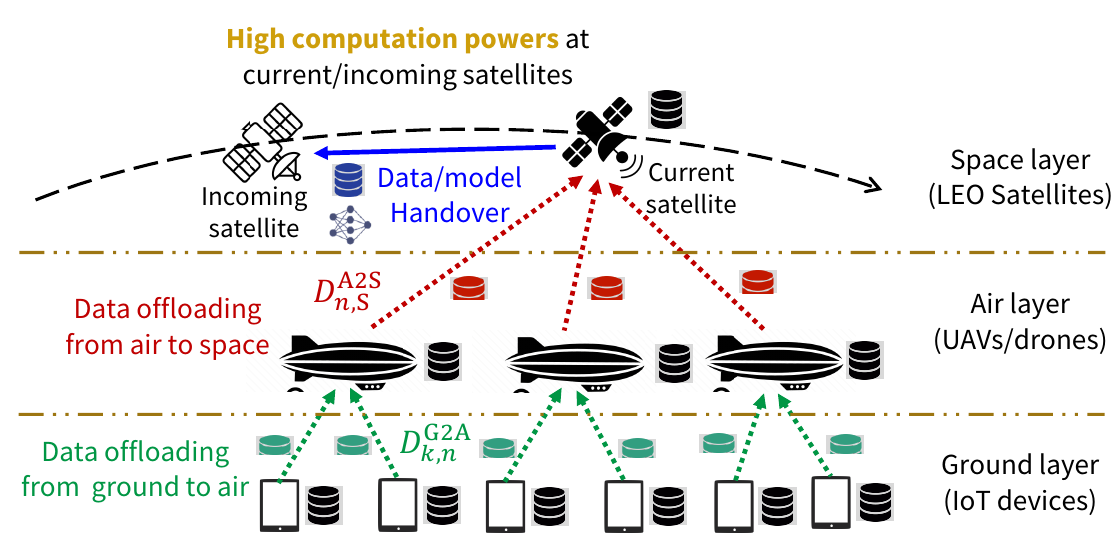} 
     \caption{High computation powers at current/incoming satellites.}\label{fig:overview_a}
 \end{subfigure}  \hspace{5mm} 
   \begin{subfigure}[b]{0.44\textwidth}
         \centering
         \includegraphics[width=\textwidth]{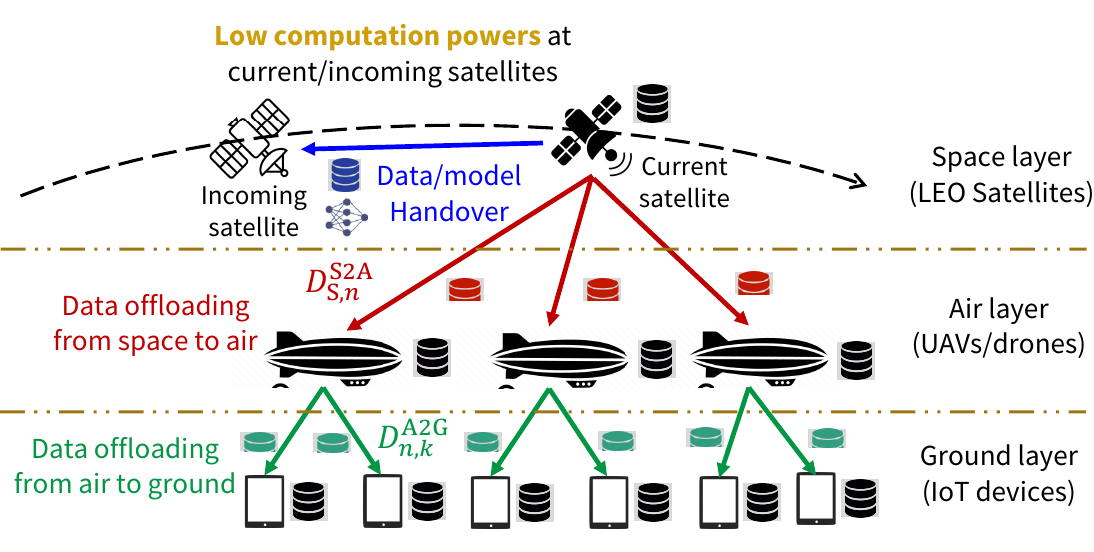} 
     \caption{Low computation powers at current/incoming satellites.}\label{fig:overview_b}
 \end{subfigure}
     \caption{Overview of adaptive data offloading/handover during FL over SAGINs, depending on the current resource availability.}
    \label{fig:overview}
    \vspace{-3mm}
  \end{figure*}

\textbf{Space-air-ground integrated networks (SAGINs):} Motivated by the potential  for providing wide wireless coverage across the Earth,   SAGINs \cite{liu2018space} have been actively studied  in the literature.   Outage analysis is conducted for SAGINs  in \cite{ye2020space}, while network control  methodologies for SAGINs  are considered  in \cite{kato2019optimizing, tang2022federated}.  In  \cite{shang2021computing, yu2021ec, liu2022energy, paul2023digital}, the authors  focused on edge computing in SAGINs, where ground devices offload their computation tasks to  the   space and air layers. 
 Compared to existing studies on SAGINs, the unique position of this work  lies in the integration of  distributed/federated ML, SAGINs, and adaptive data offloading/handover.  
 Beyond enhancing wireless coverage,  we provide additional guidelines for  intelligent ML services   in remote areas with the assistance of SAGINs.

The rest of the paper is organized as follows. We    describe the problem setup 
 in Section \ref{sec:system},  followed by an overview of    the methodology in Section \ref{sec:main}.  In Section \ref{sec:data_off}, we analyze the latency of our scheme and optimize     data offloading. Theoretical convergence results are provided in Section \ref{sec:convergence},  and numerical results are presented in Section \ref{sec:experiments}.
 Finally, we draw conclusion and future directions in Section \ref{sec:conclusion}.

\section{Problem Setup}\label{sec:system}
 
We consider a SAGIN that is composed of  space, air, and ground layers. We let $\mathcal{G}$ be the   set that consists of $K$  terrestrial devices located at a specific target region that lacks a base station.    
We denote $D_k=D_k^l\cup D_k^o$ as the local dataset of  device $k$ with $D_k^l\cap D_k^o=\emptyset$, where $D_k^l$ is the privacy-sensitive dataset that should be kept locally at each device, while $D_k^o$ consists of non-sensitive samples that can be offloaded to other nodes. We define $\alpha_k=|D_k^o|/|D_k|$ as the portion of non-sensitive data samples at ground device $k$, where  $|D|$ represents the number of  samples in dataset $D$.  This problem setting covers various applications, including (i) autonomous vehicles or mobile phones that  collect data with both non-sensitive  classes (e.g., traffic lights, trees) and sensitive classes (e.g., humans),   
(ii)   hospitals with data of   patients  who have  agreed   with the privacy policy and who have not agreed with it, (iii) sensor devices for  disaster predictions in coastal regions that mostly collect  non-sensitive samples.

In the air layer, we consider a set $\mathcal{A}$ with  $N$ air nodes (e.g., UAVs)  covering  the  target region. Each air node $n$ is associated with the device set 
$\mathcal{G}_n$, where $\mathcal{G} = \cup_{n=1}^N\mathcal{G}_n$ holds with $\mathcal{G}_{n_1}\cap \mathcal{G}_{n_2}=\emptyset$ if $n_1\neq n_2$.  In the space layer, we consider LEO satellites that are moving according to their own orbits. Each ground device can directly communicate with the corresponding air node in the air layer, while each air node can communicate with the satellite that is covering the target region. Fig. \ref{fig:overview} illustrates the overview of our system model.

The goal is to  train a shared global model $\mathbf{w}^*$ tailored to the datasets collected at  ground devices in $\mathcal{G}$. Specifically, we aim to  minimize the following objective function: 
\begin{equation}\label{original_loss}
F(\mathbf{w}) = \sum_{k=1}^K\lambda_kF_k(\mathbf{w}),
\end{equation}
where $\lambda_k=\frac{|D_k|}{\sum_{j\in \mathcal{G}}|D_j|}$ is the relative dataset size of device $k$.  $F_k(\mathbf{w})$ is the local loss function of device $k$ defined as  
 $F_k(\mathbf{w})=\frac{1}{|D_k|}\sum_{x\in D_k}\ell(x; \mathbf{w})$,  
where $\ell(x; \mathbf{w})$ is the loss (e.g., cross-entropy loss) obtained with data sample $x$ and model $\mathbf{w}$.

There are several key challenges in achieving the above goal in remote areas. First, it is difficult to aggregate the  trained models within such regions without a base station. Secondly, the terrestrial devices (e.g., IoT sensors) are often equipped with low computation capabilities,  significantly slowing down the training process.  In this work, we use space and air nodes as model aggregators to solve the first challenge, and also use them as edge computing units to process data samples offloaded from the ground layer, to tackle the second challenge. 
 
\section{Methodology Overview}\label{sec:main}
In this section, we provide an overview of our methodology that achieves the aforementioned objectives 
in SAGINs.  The proposed algorithm consists of $R$ global rounds indexed by $r=0,1,\dots, R-1$. In the following, we focus on a specific round $r$ to describe the process of our scheme.  

\subsection{Adaptive Inter-Layer Data  Offloading}\label{subsec_algo:data_offload}

  Let $D_{\textsf{G}, k}^{(r)}$, $D_{\textsf{A}, n}^{(r)}$, and $D_{\textsf{S}}^{(r)}$ denote the  local datasets  at  node $k\in \mathcal{G}$ in the ground  layer, node $n\in \mathcal{A}$ in the air layer, and the  satellite  that is currently serving the targeting region,  respectively,  at the beginning of round $r$. Note that we have
\begin{equation}
D_{\textsf{G}, k}^{(0)} = D_k, \ \forall k \in \mathcal{G}, \ 
D_{\textsf{A}, n}^{(0)} = \emptyset, \ \forall n \in \mathcal{A}, \ 
D_{\textsf{S}}^{(0)} =  \emptyset
\end{equation}
for $r=0$ since data samples are generated at the ground devices.

Depending on various system environments, 
inter-layer data offloading   is  first performed across the   network to obtain the updated datasets $D_{\textsf{G}, k}^{(r+1)}$, $D_{\textsf{A}, n}^{(r+1)}$, and $D_{\textsf{S}}^{(r+1)}$ at the nodes in each layer. Fig. \ref{fig:overview} illustrates example scenarios of  adaptive data offloading depending on the computation capabilities of the satellites. Intuitively, if the current/incoming satellites have relatively high computation powers, more data samples can be offloaded to the space layer. Otherwise, data samples should be transmitted from the space layer to other layers for load balancing.   The data offloading solution is also affected by the coverage times of the satellites over the target region.   

We describe the detailed optimization procedure for  our adaptive data offloading strategy  later   
in Section \ref{sec:data_off}, as it is built upon the analysis  provided in the following subsections.

\subsection{Local Training at Ground and Air Layers}\label{subsec_algo:local_update}
Based on the updated datasets  $D_{\textsf{G}, k}^{(r+1)}$, $D_{\textsf{A}, n}^{(r+1)}$, and $D_{\textsf{S}}^{(r+1)}$  obtained from Section \ref{subsec_algo:data_offload}, the  nodes in the system conduct local model updates. We  first describe the local training process at the ground and air layers. At the beginning of global round $r$, all nodes in the system have the synchronized model represented by  $\mathbf{w}^{(r)}$. Starting from the initial  model  $\mathbf{w}_{\textsf{G}, k}^{(r, 0)} = \mathbf{w}_{\textsf{A}, n}^{(r, 0)} =\mathbf{w}^{(r)}$, each ground device $k$ and air node $n$ updates its model for $H$ local iterations according to
\begin{equation} \label{eq:local__GG}
\hspace{-5mm}
\resizebox{0.46\textwidth}{!}{  $\mathbf{w}_{\textsf{G}, k}^{(r, h+1)} = \mathbf{w}_{\textsf{G}, k}^{(r, h)} - \eta_{\textsf{G}, k}^{(r)}\tilde{\nabla}\ell_{\textsf{G}, k}^{(r+1)}(\mathbf{w}_{\textsf{G}, k}^{(r,h)}), h=0,\ldots H\!-\!1,$}
\hspace{-4mm}
\end{equation}
\begin{equation}\label{eq:local__AA}
\hspace{-4mm}
\resizebox{0.46\textwidth}{!}{  $
\mathbf{w}_{\textsf{A}, n}^{(r,h+1)} = \mathbf{w}_{\textsf{A}, n}^{(r,h)} - \eta_{\textsf{A}, n}^{(r)}\tilde{\nabla}\ell_{\textsf{A}, n}^{(r+1)}(\mathbf{w}_{\textsf{A}, n}^{(r,h)}), h=0,\ldots H\!-\!1,
$}\hspace{-3mm}
\end{equation}
where $\mathbf{w}_{\textsf{G}, k}^{(r,h)}$ and $\mathbf{w}_{\textsf{A}, n}^{(r,h)}$
are the models after $h$  local iterations at  global round $r$, $\ell_{\textsf{G}, k}^{(r+1)}(\cdot) = \frac{1}{|D_{\textsf{G}, k}^{(r+1)}|}\sum_{x\in D_{\textsf{G}, k}^{(r+1)}} \ell(x;\cdot)$ and  $\ell_{\textsf{A}, n}^{(r+1)}(\cdot) = \frac{1}{|D_{\textsf{A}, n}^{(r+1)}|}\sum_{x\in D_{\textsf{A}, n}^{(r+1)}} \ell(x;\cdot)$ are the local loss functions defined at the corresponding nodes. Also, $\tilde{\nabla}\ell_{\textsf{G}, k}^{(r+1)}(\cdot)$ and  $\tilde{\nabla}\ell_{\textsf{A}, n}^{(r+1)}(\cdot)$ denote the computed mini-batch gradients, where the size  of the mini-batch can be  set based on the size of the local dataset.   
 $\eta_{\textsf{G}, k}^{(r)}$ and $ \eta_{\textsf{A}, n}^{(r)}$ represent the learning rates at ground device $k$ and air node $n$, respectively.

 The required local computation times (in seconds) at ground device $k$ and air node $n$  
for model updates are expressed as 
\begin{equation}\label{eq:Ground_node_local}
\tau_{\textsf{G},k}^{\textsf{local}, (r)} = \frac{m_{\textsf{G},k}|D_{\textsf{G},k}^{(r+1)}|}{f_{\textsf{G},k}}, \ \ \tau_{\textsf{A},n}^{\textsf{local}, (r)} = \frac{m_{\textsf{A},n}|D_{\textsf{A},n}^{(r+1)}|}{f_{\textsf{A},n}},  
\end{equation}
respectively, where $f_{\textsf{G},k}$, $f_{\textsf{A},n}$ are the CPU frequencies (in cycles/sec)  and $m_{\textsf{G},k}$, $m_{\textsf{A},n}$ are the numbers of required CPU cycles to update the model with one data sample (in cycles/sample) at the corresponding nodes.

\begin{figure*}[t]
  \centering
         \includegraphics[width=0.93\textwidth]{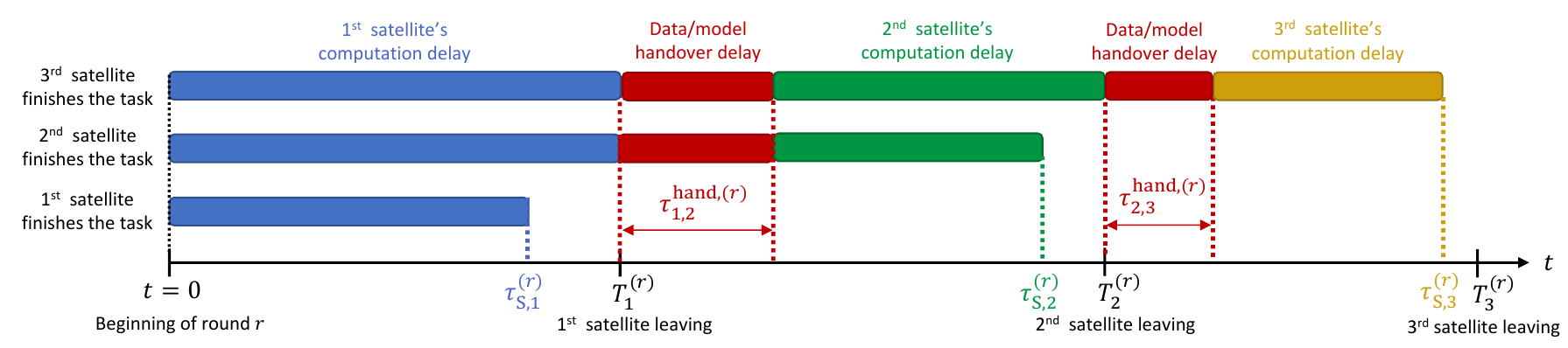} 
              \vspace{-1mm}
      \caption{Illustration of  model training and intra-layer data/model handover procedures at the space layer. If the current satellite is not able to complete the task within its coverage time over the target region, the next  incoming satellite continues local training after receiving the  dataset and the   model from the previous satellite to ensure a seamless FL process.}\label{fig:coverage_times}
   \end{figure*}

 \subsection{Satellite-Side  Training and Data/Model Handover}\label{subsec_algo:local_update_satellite}

In parallel with the local training process at the air/ground layers, the current satellite  also updates the model using dataset $D_{\textsf{S}}^{(r+1)}$.  Starting from $\mathbf{w}_{\textsf{S}}^{(r, 0)} =\mathbf{w}^{(r)}$, the model update process at the satellite can be written as follows:
 \begin{equation}\label{eq:update_sat}
\hspace{-5mm}
\resizebox{0.46\textwidth}{!}{  $\mathbf{w}_{\textsf{S}}^{(r, h+1)} = \mathbf{w}_{\textsf{S}}^{(r, h)} - \eta_{\textsf{S}}^{(r)}\tilde{\nabla}\ell_{\textsf{S}}^{(r+1)}(\mathbf{w}_{\textsf{S}}^{(r,h)} ),  h=0,\ldots H\!-\!1,$}
\hspace{-3mm}
\end{equation}
where $\tilde{\nabla}\ell_{\textsf{S}}^{(r+1)}(\cdot)$ and $\eta_{\textsf{S}}^{(r)}$ are the satellite-side  stochastic mini-batch gradient and learning rate, respectively. The size of  the mini-batch is set to   $|D_{\textsf{S}}^{(r+1)}|/H$ so that   all data samples in  $D_{\textsf{S}}^{(r+1)}$ can be processed in $H$ iterations.

\textbf{Data/model handover:} 
 In satellite networks, satellites are perceived as non-stationary units, where at each snapshot of the network  each LEO satellite covers a different region
compared to other LEO satellites and may have its own specific task tailored to its coverage area (e.g., edge computing, FL, or communication services). In our setup, the current satellite that is covering the target area is responsible for conducting FL over that region. However, a key challenge is that each satellite has a limited coverage time over the target region due to the mobility.  This motivates us to consider an intra-layer data/model handover strategy within the space layer, to ensure a seamless FL process. Specifically, the current satellite transmits the updated model and dataset to the  incoming satellite before leaving the target region, so that this new satellite can continue  model  training in the space layer using  dataset $D_{\textsf{S}}^{(r+1)}$ during its coverage period over the target region.  
These local training and  handover  steps are repeated  
 until all data samples in $D_{\textsf{S}}^{(r+1)}$ are processed, based on  a series of incoming satellites that will cover the target region.

The  handover delay between the $i$-th and $(i+1)$-th satellites at global round $r$ can be written as follows:
\begin{equation}
\tau_{i,i+1}^{\textsf{hand}, (r)} = \frac{Q(\mathbf{w}) + q|D_{\textsf{S}}^{(r+1)}|}{Z_{i, i+1}^{\textsf{ISL},(r)}},
\end{equation}
where  $Q(\mathbf{w})$ is the model size (in bits),  $q$ is the size of each data sample (in bits) and $Z_{i, i+1}^{\textsf{ISL}}$ is the transmission rate for  inter-satellite link (ISL) communications between $i$-th and $(i+1)$-th satellites. Referring to \cite{leyva2021inter, Razmi1}, we have $Z_{i, i+1}^{\textsf{ISL},(r)} = B\log_2(1+\frac{p_{\textsf{S}, i}^{(r)}A_i^{\text{Tx}}A_{i+1}^{\text{Rx}}}{C_{i, i+1}N_0})$, where $B$ is the bandwidth, $p_{\textsf{S}, i}^{(r)}$ is the transmit power of the $i$-th satellite, $A_{i}^{\text{Tx}}$ and $A_{i+1}^{\text{Rx}}$  are  the Tx and Rx gains of antenna,  $C_{i, i+1}$ is the free space path loss between satellites, $N_0$ is the noise power density.

\textbf{Latency at the space layer:} Based on the above data/model handover strategy, we now characterize the  training latency at the space layer. Let   $f_{\textsf{S},i}^{(r)}$ represent the CPU frequency of the $i$-th satellite covering the target region at global round $r$. We also denote $m_{\textsf{S},i}^{(r)}$ as the number of CPU cycles required to process one data sample at the $i$-th satellite at round $r$.  Moreover, let $T_i^{(r)}$ denote the delay until the $i$-th   satellite  leaves  the coverage of the target region,  measured from the moment when global round $r$ has started. Trivially, for the satellites that do not join or leave the region in round $r$, $T_i^{(r)}$ becomes infinity.

To gain insights, we start  with some examples illustrated in Fig. \ref{fig:coverage_times}. Suppose that the first satellite  is able to process   the whole dataset $D_{\textsf{S}}^{(r+1)}$ within the time duration $T_1^{(r)}$. Then, the local computation delay  $\tau_{\textsf{S}, 1}^{(r)}$ (in seconds) at the space layer  
 can be written as follows:
\begin{equation}\label{eq:S__1}
\tau_{\textsf{S}, 1}^{(r)} = {m_{\textsf{S},1}^{(r)}|D_{\textsf{S}}^{(r+1)}|}\big/{f_{\textsf{S},1}^{(r)}}. 
\end{equation}

  However, if $ \tau_{\textsf{S}, 1}^{(r)} > T_1^{(r)}$, indicating that the first satellite is unable to  complete the computation before leaving the target region,   data/model handover from the first satellite to the second satellite  is conducted.   Note that the number of data samples that can be processed at the first satellite within time duration  $T_1^{(r)}$ is $(f_{\textsf{S},1}^{(r)}/m_{\textsf{S},1}^{(r)})T_1^{(r)}$. 
 Hence,  the amount of data samples that should be processed at the satellites other than the first one becomes $|D_{\textsf{S}}^{(r+1)}| - (f_{\textsf{S},1}^{(r)}/m_{\textsf{S},1}^{(r)})T_1^{(r)}$. Now suppose that the second satellite can process all $|D_{\textsf{S}}^{(r+1)}| - (f_{\textsf{S},1}^{(r)}/m_{\textsf{S},1}^{(r)})T_1^{(r)}$ data samples before leaving the target region. Then the computation time at the second satellite to finish local training can be expressed as $ m_{\textsf{S},2}^{(r)}(|D_{\textsf{S}}^{(r+1)}| - \frac{f_{\textsf{S},1}^{(r)}}{m_{\textsf{S},1}^{(r)}}T_1^{(r)})/f_{\textsf{S},2}^{(r)}$. This leads to the following latency result:
     \begin{align}\label{eq:S__2}
              \tau_{\textsf{S}, 2}^{(r)} = T_1^{(r)}+  \tau_{1,2}^{\textsf{hand}, (r)}  +  \frac{m_{\textsf{S},2}^{(r)}(|D_{\textsf{S}}^{(r+1)}| - \frac{f_{\textsf{S},1}^{(r)}}{m_{\textsf{S},1}^{(r)}}T_1^{(r)})}{f_{\textsf{S},2}^{(r)}}.
             \end{align}
The result in (\ref{eq:S__2}) incorporates the computation time at the first satellite, i.e., $T_1^{(r)}$, the  handover delay, i.e., $\tau_{1,2}^{\textsf{hand}, (r)}$, and the computation time at the second satellite, i.e., the last term.

However, if $ \tau_{\textsf{S}, 2}^{(r)} > T_2^{(r)}$, the local training cannot be completed before the second satellite leaves the  target area. In this case, the third satellite processes the remaining data after receiving the information from the second satellite via ISL communication. 
  Overall, we obtain the following result:
\begin{align}\label{eq:local_satellite_time}
 \tau_{\textsf{S}}^{(r)} = 
\begin{dcases}  
\tau_{\textsf{S}, 1}^{(r)},  &    \text{if $\tau_{\textsf{S}, 1}^{(r)}< T_1^{(r)}$ ($1^{\text{st}}$ satellite finishes the task)}\\
\tau_{\textsf{S},2}^{(r)},  &  \text{if $\tau_{\textsf{S}, 2}^{(r)}< T_2^{(r)}$ ($2^{\text{nd}}$ satellite  finishes the task)}\\
\tau_{\textsf{S}, 3}^{(r)},  &   \text{if $\tau_{\textsf{S}, 3}^{(r)}< T_3^{(r)}$ ($3^{\text{rd}}$ satellite  finishes the task)}\\
&\vdotswithin{}
\end{dcases}
\end{align}
where  $\tau_{\textsf{S}, 1}^{(r)}$ and $\tau_{\textsf{S}, 2}^{(r)}$ are defined in (\ref{eq:S__1}) and (\ref{eq:S__2}) while $\tau_{\textsf{S}, 3}^{(r)}$ is written as follows:

{\small
\begin{align}\label{qpqpqpqqpqp}
&\tau_{\textsf{S}, 3}^{(r)}=T_2^{(r)}+ \tau_{2,3}^{\textsf{hand}, (r)} \\  &+ \frac{m_{\textsf{S},3}^{(r)}\Big(|D_{\textsf{S}}^{(r+1)}| - \frac{f_{\textsf{S},1}^{(r)}}{m_{\textsf{S},1}^{(r)}}T_1^{(r)}   - \frac{f_{\textsf{S},2}^{(r)}}{m_{\textsf{S},2}^{(r)}}(T_2^{(r)} - T_1^{(r)} - \tau_{1,2}^{\textsf{hand}, (r)}) \Big)}{f_{\textsf{S},3}^{(r)}}.\nonumber
\end{align}}
As illustrated in Fig. \ref{fig:coverage_times}, the term $T_2^{(r)}$ in (\ref{qpqpqpqqpqp}) captures the delay until the second satellite leaves the target region, $\tau_{2,3}^{\textsf{hand}, (r)}$ is the handover delay, and the last term is the delay for the third satellite to complete the remaining tasks. For an arbitrary $i\geq 2$, we can generalize the result as follows:
  \begin{align}
\tau_{\textsf{S}, i}^{(r)}=T_{i-1}^{(r)}+ \tau_{i-1,i}^{\textsf{hand}, (r)} + \frac{m_{\textsf{S},i}^{(r)}(|D_{\textsf{S}}^{(r+1)}| - \Omega_i^{(r)})}{f_{\textsf{S},i}^{(r)}},
\end{align}
where $\Omega_i^{(r)}$ is the amount of data samples processed  prior to the $i$-th satellite at round $r$. Fig. \ref{fig:coverage_times} summarizes the idea of the repeated local training and data/model handover processes at the space layer.

 \subsection{Model Aggregation}\label{subsec_algo:model_agg}

After  local updates are completed according to Sections  \ref{subsec_algo:local_update} and \ref{subsec_algo:local_update_satellite},  model aggregation is conducted to obtain a new  global model. Specifically, each air node $n$ aggregates the models $\{\mathbf{w}_{\textsf{G},k}^{(r+1)}\}_{k\in\mathcal{G}_n}$ sent from the ground devices in its coverage and the model $\mathbf{w}_{\textsf{A},n}^{(r+1)}$ trained by its own. Then,  each air node $n$ sends the aggregated model to the current satellite for global aggregation. The final global model becomes

 {\small
\begin{equation} 
\mathbf{w}^{(r+1)} =  \sum_{k\in \mathcal{G}}\lambda_{\textsf{G},k}^{(r+1)}\mathbf{w}_{\textsf{G}, k}^{(r, H)} + \sum_{n\in \mathcal{A}}\lambda_{\textsf{A},n}^{(r+1)}\mathbf{w}_{\textsf{A}, n}^{(r, H)}+ \lambda_{\textsf{S}}^{(r+1)} \mathbf{w}_{\textsf{S}}^{(r, H)},
\end{equation}}
where $\lambda_{\textsf{G},k}^{(r+1)} = \frac{|D_{\textsf{G},k}^{(r+1)}|}{\sum_{j\in \mathcal{G}}|D_j|}$, $\lambda_{\textsf{A},n}^{(r+1)} = \frac{|D_{\textsf{A},n}^{(r+1)}|}{\sum_{j\in \mathcal{G}}|D_j|}$,  and $\lambda_{\textsf{S}}^{(r+1)} = \frac{|D_{\textsf{S}}^{(r+1)}|}{\sum_{j\in \mathcal{G}}|D_j|}$ are the portions of data samples  at each node.

The  delay for uploading the model  from ground  device $k$ to air node $n$ 
can be written as  follows:
\begin{equation}\label{eq:model_trabnsmit}
\tau_{k,n}^{\textsf{\textsf{G2A}},  (r)} = \frac{Q(\mathbf{w})}{Z_{k,n}^{\textsf{G2A},(r)}},
\end{equation}
where $Z_{k,n}^{\textsf{G2A},(r)}$ is the uplink communication rate between ground device $k$ and air node $n$  expressed as\footnote{In scenarios where instantaneous channel is available via feedback, the latency can be written without the expectation.}
\begin{equation}
Z_{k,n}^{\textsf{G2A},(r)}=\mathbb{E}\Big [b_{k,n}^{(r)}\log_2(1+\frac{p_{\textsf{G}, k}|h_{k,n}^{(r)}|^2}{b_{k,n}^{(r)}N_0})\Big ].
\end{equation}
Here, $p_{\textsf{G}, k}$ is the transmit power of ground device $k$, $b_{k,n}^{(r)}$ is the bandwidth, and $h_{k,n}^{(r)} = \beta_0/(d_{k,n}^{(r)})^{\gamma^{\textsf{G2A}}}g$ is the channel between device $k$ and air node $n$, which is defined with the distance $d_{k,n}^{(r)}$, pathloss exponent between ground and air $\gamma^{\textsf{G2A}}$,     channel gain  $\beta_0$ at the reference distance
of 1 meter, and Rayleigh fading parameter $g$. Similarly, we can also define the model upload delay from air node $n$ to the current satellite, i.e.,  $\tau_{n,\textsf{S}}^{\textsf{A2S}, (r)}$,  based on the communication rate $Z_{n,\textsf{S}}^{\textsf{A2S},(r)}$ between air node $n$ and the current satellite covering the target region\footnote{Following \cite{pang2021irs, guo2022distributed,  callegaro2021optimal}, Rayleigh fading can be adopted between the ground device and the air node, considering obstacles in remote areas such as forests and mountainous regions.  In scenarios where the line-of-sight link is dominant, we can use the free-path space loss model by setting $h_{k,n}^{(r)} = \beta_0/(d_{k,n}^{(r)})^{2}$ as in \cite{fu2023federated, wu2018joint}}.

\section{Adaptive Data Offloading Optimization} \label{sec:data_off}
 
In this section, we provide details for our    data offloading step outlined in Section \ref{subsec_algo:data_offload}. This process aims to construct    $\{D_{\textsf{G}, k}^{(r+1)}\}_{k\in\mathcal{G}}$, $\{D_{\textsf{A}, n}^{(r+1)}\}_{n\in\mathcal{A}}$, and  $D_{\textsf{S}}^{(r+1)}$ from $\{D_{\textsf{G}, k}^{(r)}\}_{k\in\mathcal{G}}$, $\{D_{\textsf{A}, n}^{(r)}\}_{n\in\mathcal{A}}$,  and $D_{\textsf{S}}^{(r)}$, at the beginning of global round $r$.     
    \vspace{-1mm}

\subsection{Characterization of Data Transmission Direction}

\textbf{Latency without data offloading:} The first step of our approach is to characterize the direction of data transmission. We start  by  deriving  the latency   without data offloading, to see which layer causes more delay. When data offloading is not considered, the overall delay at round $r$ can be written as 
\begin{equation}\label{Lat_wo_offload1}
\tau^{(r)}= \max\Big\{\tau_{S}^{(r)}, \max_{n\in \mathcal{A}}\{\tau_{\textsf{A},n}^{(r)}+\tau_{n,\textsf{S}}^{\textsf{A2S},(r)}\}  \Big\},
\end{equation}
where $\tau_{S}^{(r)}$ is the completion time at the space layer defined in (\ref{eq:local_satellite_time}) and  $\tau_{n,\textsf{S}}^{\textsf{A2S},(r)}$ is the model transmission delay  from   air node $n$ to the current satellite, similar to (\ref{eq:model_trabnsmit}). 
 $\tau_{\textsf{A},n}^{(r)}$  is  the delay until air node $n$ aggregates its own updated model with the models sent from the   devices in its coverage area $\mathcal{G}_n$: 
 \begin{equation}\label{Lat_wo_offload2}
\tau_{\textsf{A},n}^{ (r)} = \max\Big\{\tau_{\textsf{A},n}^{\textsf{local}, (r)}, \max_{k\in \mathcal{G}_n}\{\tau_{\textsf{G},k}^{\textsf{local}, (r)}+ \tau_{k,n}^{\textsf{G2A}, (r)} \}\Big\}.
\end{equation}  
Here, $\tau_{\textsf{A},n}^{\textsf{local}, (r)}$ and $\tau_{\textsf{G},k}^{\textsf{local}, (r)}$  are the local computation times at air node $n$   and ground device $k$, respectively, as described in (\ref{eq:Ground_node_local}).
Here, we note that all notations in (\ref{Lat_wo_offload1}) and (\ref{Lat_wo_offload2}) are defined with the datasets  before data offloading, i.e., $\{D_{\textsf{G}, k}^{(r)}\}_{k\in\mathcal{G}}$, $\{D_{\textsf{A}, n}^{(r)}\}_{n\in\mathcal{A}}$, $D_{\textsf{S}}^{(r)}$, to characterize the data offloading direction.

\textbf{Data transmission scenarios:} Our adaptive data offloading method is motivated by the dynamic nature   of SAGINs, including the computation   capabilities as well as the coverage times of current/incoming satellites.  
We consider two different scenarios depending on the direction of data transmission.

\textit{\underline{(i) Case I:}}   
$\tau_{S}^{(r)} > \max_{n\in \mathcal{A}}\{\tau_{\textsf{A},n}^{(r)}+\tau_{n,\textsf{S}}^{\textsf{A2S},(r)}\}$\textit{(Offloading from space to air/ground)}\textbf{.} 
 Case I  considers the scenario where   the current and the next few incoming satellites have relatively  low computation/communication capabilities. In this case, we allow 
data samples  to be  transmitted from  the space layer to air/ground layers for load balancing.

\textit{\underline{(ii) Case II:}}  
$\tau_{S}^{(r)}<\max_{n\in \mathcal{A}}\{\tau_{\textsf{A},n}^{(r)}+\tau_{n,\textsf{S}}^{\textsf{A2S},(r)}\}$ \textit{(Offloading from   air/ground to space)}\textbf{.} 
In this case, the current/incoming satellites have relatively large computation powers. Hence, we propose data transmission   from air/ground layers to the space layer for load balancing.

\textbf{Objective:} Our objective is to adaptively optimize  data offloading across space-air-ground layers   to minimize the  latency. 
By incorporating the data offloading delay, we can rewrite the overall latency in (\ref{Lat_wo_offload1}) into the following form:
\begin{equation}\label{Lat_with_lllof}
\bar{\tau}^{(r)}:= \max\Big\{\bar{\tau}_{S}^{(r)}, \max_{n\in \mathcal{A}}\{\bar{\tau}_{\textsf{A},n}^{(r)}+\tau_{n,\textsf{S}}^{\textsf{A2S},(r)}\}  \Big\}.
\end{equation}
In (\ref{Lat_with_lllof}),  $\bar{\tau}_{S}^{(r)}$ is the  new delay at the space layer and
\begin{equation}\label{eq_sdf}
\bar{\tau}_{\textsf{A},n}^{(r)} := \max\Big\{\bar{\tau}_{\textsf{A},n}^{\textsf{local},(r)}, \max_{k\in\mathcal{G}_n}\{\bar{\tau}_{\textsf{G},k}^{\textsf{local}, (r)} + \tau_{k,n}^{\textsf{G2A}, (r)}\} \Big\}
\end{equation}
is the  new completion time at air node $n$ until   all the models in its coverage   are aggregated, considering data offloading. $\bar{\tau}_{\textsf{A},n}^{\textsf{local},(r)}$ and $\bar{\tau}_{\textsf{G},k}^{\textsf{local}}$ are the updated delays to finish local training at air node $n$ and ground device $k$, respectively, under this data offloading framework.

In the following, we will characterize   the new delays $\bar{\tau}_{S}^{(r)}$, $\bar{\tau}_{\textsf{A},n}^{\textsf{local},(r)}$, and $\bar{\tau}_{\textsf{G},k}^{\textsf{local}}$  in (\ref{Lat_with_lllof}) and (\ref{eq_sdf}) by considering data offloading.  Then, we will   optimize the amount of data being offloaded across the   layers to minimize  $\bar{\tau}^{(r)}$.

\subsection{Case I:    Data Offloading  from  Space  to Air/Ground}
We first consider  Case I.
Let $D_{\textsf{S},n}^{\textsf{S2A}, (r)}$ be the dataset sent from the space layer to air node $n$ in the air layer. 

 \textbf{Dataset  and latency characterization at the space layer:} Then the updated dataset  $D_{\textsf{S}}^{ (r+1)}$ at the space layer after data offloading satisfies the following criterion:
\begin{equation}\label{eq:sate_dataset_updated}
|D_{\textsf{S}}^{ (r+1)}| = |D_{\textsf{S}}^{ (r)}| - \sum_{n\in\mathcal{A}}|D_{\textsf{S},n}^{\textsf{S2A}, (r)}|.
\end{equation}
Accordingly, we can obtain the updated satellite-side delay $\bar{\tau}_{\textsf{S}}^{(r)}$    by inserting $|D_{\textsf{S}}^{ (r+1)}| = |D_{\textsf{S}}^{ (r)}| - \sum_{n\in\mathcal{A}}|D_{\textsf{S},n}^{\textsf{S2A}, (r)}|$ to (\ref{eq:local_satellite_time}). In   (\ref{eq:sate_dataset_updated}), $\{|D_{\textsf{S},n}^{\textsf{S2A}, (r)}|\}_{n\in\mathcal{A}}$ is the set of   parameters that we would like to optimize.  We also aim to optimize the load balancing between air and ground layers. To achieve this, we will first 
study the load balancing between air node $n$ and the associated ground devices  in $\mathcal{G}_n$ when  $|D_{\textsf{S},n}^{\textsf{S2A}, (r)}|$ is given. After that, we focus on  the load balancing between the space and air layers.

We  first characterize the direction of data  transmission between the air and ground.    If (i) $|D_{\textsf{S},n}^{\textsf{S2A}, (r)}|$ is provided from the space layer to  air node $n$, and (ii)   data offloading between air and ground layers is not performed,   the local computation delay at air node $n$ can be rewritten as follows: 
\begin{equation}\label{eqeee}
\tau_{\textsf{A},n}^{\textsf{local},(r)} = \max\{\frac{m_{\textsf{A},n}|D_{\textsf{A},n}^{(r)}|}{f_{\textsf{A},n}},   \frac{q|D_{\textsf{S},n}^{\textsf{S2A}, (r)}|}{Z_{\textsf{S},n}^{\textsf{S2A},(r)}}\}+ \frac{m_{\textsf{A},n}|D_{\textsf{S},n}^{\textsf{S2A}, (r)}|}{f_{\textsf{A},n}}.
\end{equation}
The result in (\ref{eqeee}) can be interpreted as follows. 
At the beginning of round $r$, the current satellite transmits dataset $D_{\textsf{S},n}^{\textsf{S2A}, (r)}$ to air node $n$.  This incurs delay of $\frac{q|D_{\textsf{S},n}^{\textsf{S2A}, (r)}|}{Z_{\textsf{S},n}^{\textsf{S2A}, (r)}}$, where $Z_{\textsf{S},n}^{\textsf{S2A}, (r)}$ is the downlink communication rate between the current satellite and air node $n$. In parallel, air node $n$ conducts local update based on the dataset  $D_{\textsf{A},n}^{(r)}$, causing delay of $\frac{m_{\textsf{A},n}|D_{\textsf{A},n}^{(r)}|}{f_{\textsf{A},n}}$. When both processes are completed, air node $n$ can update the model using dataset $D_{\textsf{S},n}^{\textsf{S2A}, (r)}$ received from the satellite, which is captured in the last term of (\ref{eqeee}).

Now if $\tau_{\textsf{A},n}^{\textsf{local}, (r)} > \max_{k\in \mathcal{G}_n}\{\tau_{\textsf{G},k}^{\textsf{local}, (r)}+ \tau_{k,n}^{\textsf{G2A}, (r)} \}$, i.e., if the computation  time at air node $n$ is larger than the completion time at the ground layer in its associated region,  we let air node $n$ transmit data samples to the ground layer for load balancing. Otherwise, i.e., if 
$\tau_{\textsf{A},n}^{\textsf{local}, (r)} < \max_{k\in \mathcal{G}_n}\{\tau_{\textsf{G},k}^{\textsf{local}, (r)}+ \tau_{k,n}^{\textsf{G2A}, (r)} \}$,   we let air node   $n$  receive data samples from the corresponding ground devices for load balancing. In the following, we describe our method assuming $\tau_{\textsf{A},n}^{\textsf{local}, (r)} > \max_{k\in \mathcal{G}_n}\{\tau_{\textsf{G},k}^{\textsf{local}, (r)}+ \tau_{k,n}^{\textsf{G2A}, (r)} \}$, where the result for the second case can be   obtained in a similar way.

 \textbf{Dataset  and latency characterization at air/ground layers:} We  define $D_{n,k}^{\textsf{A2G}, (r)}$ as  the dataset that is sent from   air node $n$ to ground device $k\in \mathcal{G}_n$ at global round $r$.  Then, the following holds for the updated dataset $D_{\textsf{A},n}^{ (r+1)}$ at air node $n$:
\begin{equation}
|D_{\textsf{A},n}^{ (r+1)}| = |D_{\textsf{A},n}^{ (r)}| +    |D_{\textsf{S},n}^{\textsf{S2A}, (r)}|  - \sum_{k\in\mathcal{G}_n}|D_{ n,k}^{\textsf{A2G}, (r)}|,
\end{equation}
which is obtained after receiving $|D_{\textsf{S},n}^{\textsf{S2A}, (r)}|$ samples from the satellite and sending $\sum_{k\in\mathcal{G}_n}|D_{ n,k}^{\textsf{A2G}, (r)}|$ samples to the ground devices in $\mathcal{G}_n$. For each ground device $k\in\mathcal{G}_n$, we can write 
\begin{equation}\label{khgmx2}
|D_{\textsf{G},k}^{ (r+1)}| = |D_{\textsf{G},k}^{ (r)}|
+ |D_{n,k}^{\textsf{A2G}, (r)}|,
\end{equation}
after receiving data from the corresponding air node $n$.

From the above definitions on the updated datasets at each layer, we can  write  $\bar{\tau}_{\textsf{A}, n}^{\textsf{local}, (r)}$ in (\ref{eq_sdf}), which represents the delay for air node $n$ to finish  computation,  as follows:  
\begin{align}\label{eq:lddlld}
&\bar{\tau}_{\textsf{A}, n}^{\textsf{local}, (r)} = \\
&\begin{dcases}  
\frac{m_{\textsf{A},n}|D_{\textsf{A},n}^{ (r+1)}|}{f_{\textsf{A},n}},  \  \  \text{if $|D_{\textsf{A},n}^{(r+1)}| \leq |D_{\textsf{A},n}^{(r)}|$}   \\
\max\Big\{\frac{m_{\textsf{A},n}|D_{\textsf{A},n}^{(r)}|}{f_{\textsf{A},n}},   \frac{q|D_{\textsf{S},n}^{\textsf{S2A}, (r)}|}{Z_{\textsf{S},n}^{\textsf{S2A},(r)}}\Big\} \\ \ \  + \frac{m_{\textsf{A},n}(|D_{\textsf{S},n}^{\textsf{S2A}, (r)}| - \sum_{k\in\mathcal{G}_n}|D_{ n,k}^{\textsf{A2G}, (r)}|)}{f_{\textsf{A},n}}, \ \text{otherwise} \nonumber
\end{dcases}
\end{align}
In (\ref{eq:lddlld}), if $|D_{\textsf{A},n}^{(r+1)}| \leq |D_{\textsf{A},n}^{(r)}|$, air node $n$ can finish computation  without waiting for 
  dataset  $D_{\textsf{S},n}^{\textsf{S2A}, (r)}$ from the satellite. On the other hand,  if $|D_{\textsf{A},n}^{(r+1)}| > |D_{\textsf{A},n}^{(r)}|$, it indicates that air node  $n$ also needs to   process   data samples received from the satellite.
  For both cases, air node  $n$ transmits $|D_{n,k}^{\textsf{A2G}, (r)}|$ data samples to ground device $k$ after  receiving data from the satellite.

Hence, for  ground device $k$, we can   write the completion time in (\ref{eq_sdf}) as    follows:
 
{\small
\begin{align}\label{qiiien}
\bar{\tau}_{\textsf{G},k}^{\textsf{local}, (r)} &= \max\Big\{\underbrace{\frac{m_{\textsf{G},k}|D_{\textsf{G},k}^{(r)}|}{f_{\textsf{G},k}}}_{\text{Comp. with original data}}, \underbrace{\frac{q|D_{\textsf{S},n}^{\textsf{S2A},(r)}|}{Z_{\textsf{S},n}^{\textsf{S2A},(r)}}+ \frac{q|D_{n,k}^{\textsf{A2G},(r)}|}{Z_{n,k}^{\textsf{A2G},(r)}}}_{\text{Comm.   for receiving data samples}} \Big\} \nonumber \\ 
&+ \underbrace{\frac{m_{\textsf{G},k}|D_{n,k}^{\textsf{A2G}, (r)}|}{f_{\textsf{G},k}}.}_{\text{Comp. with   received data from air node $n$}} 
\end{align}
}
Specifically, each ground device   starts computation with its original data when  round $r$ begins, and in parallel, waits until data samples from  air node $n$ arrives. Then, each device finishes computation using  data  samples received  from air node $n$.

	\begin{algorithm}[t]
 	\caption{Load Balancing Between Air Node $n$ and the Ground Devices in $\mathcal{G}_n$  }\label{algo:Air_Ground}
	\small
		\begin{algorithmic}[1]
		\STATE \textbf{Input:}  $\nu_{L,1}=\nu_{L,2}=0$, an appropriate $\nu_{U_1}$, $\nu_{U_2}$, and  small $\epsilon_1$, $\epsilon_2$. Initialized $|D_{n,k}^{\textsf{A2G}, (r)}|=0$ for all $k\in \mathcal{G}_n$.  Fixed $|D_{\textsf{S},n}^{\textsf{S2A}, (r)}|$. \\
	\STATE \textbf{Output: } Optimal data allocation  $\{|D_{n,k}^{\textsf{A2G}, (r)}|\}_{k\in\mathcal{G}_n}$ between air node $n$ and ground devices in $\mathcal{G}_n$. \\
			 \WHILE{$\nu_{U, 1}-\nu_{L,1}\geq \epsilon_1$}
			 		\STATE{Set $Y_n= (\nu_{U, 1}+ \nu_{L,1})/2$}
			 \STATE{Obtain $\{|D_{n,k}^{\textsf{A2G}, (r)}|\}_{k\in \mathcal{G}_n}$ based on the following while loop:}
	 \STATE{Set an appropriate $\nu_{L_2}$ and $\nu_{U_2}$.}	\WHILE{$\sum_{k\in\mathcal{G}_n}|D_{n,k}^{\textsf{A2G},(r)}|<(1-\epsilon_2)Y_{n}$ or $\sum_{k\in\mathcal{G}_n}|D_{n,k}^{\textsf{A2G},(r)}|> (1+\epsilon_2)Y_n$}
		\FOR{each $k\in \mathcal{G}_n$}
		\STATE{Compute $|D_{n,k}^{\textsf{A2G},(r)}|$  to make $\bar{\tau}_{\textsf{G},k}^{\textsf{local}, (r)} + \tau_{k,n}^{\textsf{G2A}, (r)}$ in (\ref{qiiien}) and  $\frac{1}{2}(\nu_{U, 2}+\nu_{L,2})$ as close as possible within  range $|D_{n,k}^{\textsf{A2G}, (r)}|\in [0, \min\{|D_{\textsf{A},n}^{(r)}|, Y_n\}]$  using bisection search.}
		\ENDFOR
		\IF{$\sum_{k\in \mathcal{G}_n}|D_{n,k}^{\textsf{A2G}, (r)}|\leq (1-\epsilon_2)Y_n$}
		\STATE{$\nu_{L, 2} \leftarrow \frac{1}{2}(\nu_{U, 2}+\nu_{L,2})$}
		\ELSE
		\STATE{$\nu_{U, 2} \leftarrow \frac{1}{2}(\nu_{U, 2}+\nu_{L,2})$}
		\ENDIF
		\ENDWHILE
				 \STATE{Compute $\max_{k\in\mathcal{G}_n}\{\bar{\tau}_{\textsf{G},k}^{\textsf{local}, (r)} + \tau_{k,n}^{\textsf{G2A}, (r)}\}$ based on (\ref{qiiien}) 
 and the obtained $\{|D_{n,k}^{\textsf{A2G}, (r)}|\}_{k\in \mathcal{G}_n}$}
 \STATE{Compute $\bar{\tau}_{\textsf{A},n}^{\textsf{local}, (r)}$ according to (\ref{eq:lddlld})}				 \STATE{\textbf{if} $\bar{\tau}_{\textsf{A},n}^{\textsf{local}, (r)} \geq \max_{k\in\mathcal{G}_n}\{\bar{\tau}_{\textsf{G},k}^{\textsf{local}, (r)} + \tau_{k,n}^{\textsf{G2A}, (r)}\}$,  set $\nu_{L,1} = Y_n$.}
				 		\STATE{\textbf{else} set $\nu_{U,1} = Y_n$.}
	 			\ENDWHILE
	\end{algorithmic}
\end{algorithm}

\setlength{\textfloatsep}{12pt}

 \textbf{Load balancing between air/ground layers:}   
 For load balancing between air and ground layers, we first  optimize $\{|D_{n,k}^{\textsf{A2G}, (r)}|\}_{k\in \mathcal{G}_n}$ that minimizes  $\bar{\tau}_{\textsf{A},n}^{(r)}$ 
in (\ref{eq_sdf}), by solving
\begin{equation}\label{subproblem_alhphaskdk}
\hspace{-4.9mm}
  \min_{\{|D_{n,k}^{\textsf{A2G}, (r)}|\}_{k\in \mathcal{G}_n}}
  \hspace{-2mm}
  \max\Big\{\bar{\tau}_{\textsf{A},n}^{\textsf{local},(r)}, \max_{k\in\mathcal{G}_n}\{\bar{\tau}_{\textsf{G},k}^{\textsf{local}, (r)} + \tau_{k,n}^{\textsf{G2A}, (r)}\} \Big\} 
  \hspace{-5mm}
 \end{equation} 
when $|D_{\textsf{S},n}^{\textsf{S2A}, (r)}|$  is given.  Note that the completion time at the ground layer, i.e.,  $\max_{k\in\mathcal{G}_n}\{\bar{\tau}_{\textsf{G},k}^{\textsf{local}, (r)} + \tau_{k,n}^{\textsf{G2A}, (r)}\}$, is an increasing function of $|D_{n,k}^{\textsf{A2G}, (r)}|$ while the computation delay at the air layer, i.e.,  $\bar{\tau}_{\textsf{A},n}^{\textsf{local}}$, is  a decreasing function of $|D_{n,k}^{\textsf{A2G}, (r)}|$. Hence, as described in Algorithm \ref{algo:Air_Ground}, we can use bisection search to make  $\bar{\tau}_{\textsf{A},n}^{\textsf{local}, (r)}$ and $\max_{k\in\mathcal{G}_n}\{\bar{\tau}_{\textsf{G},k}^{\textsf{local}, (r)} + \tau_{k,n}^{\textsf{G2A}, (r)}\}$ as close as possible, by controlling our optimization parameters $\{|D_{n,k}^{\textsf{A2G}, (r)}|\}_{k\in \mathcal{G}_n}$. In  Algorithm \ref{algo:Air_Ground}, we  first solve
\begin{align}\label{subproblem_alhpha}
 &  \min_{\{|D_{n,k}^{\textsf{A2G}, (r)}|\}_{k\in \mathcal{G}_n}}\max\Big\{\bar{\tau}_{\textsf{A},n}^{\textsf{local},(r)}, \max_{k\in\mathcal{G}_n}\{\bar{\tau}_{\textsf{G},k}^{\textsf{local}, (r)} + \tau_{k,n}^{\textsf{G2A}, (r)}\} \Big\}\nonumber \\  
& \text{subject to:}   
  \sum_{k\in \mathcal{G}_n}|D_{n,k}^{\textsf{A2G}, (r)}|=Y_n 
 \end{align}
 for a given $Y_n$, and then optimize $Y_n$ to minimize $\bar{\tau}_{\textsf{A},n}^{(r)}$ in (\ref{eq_sdf}), by implementing bisection search in a hierarchical way.

\begin{algorithm}[t]
 	\caption{Load Balancing Across Space-Air-Ground Layers }\label{algo:overall}
		\small
		\begin{algorithmic}[1]
		\STATE \textbf{Input: }$\nu_{L,1}=\nu_{L,2}=0$, an appropriate $\nu_{U_1}$, $\nu_{U_2}$, and  small $\epsilon_1$, $\epsilon_2$. Initialized $|D_{\textsf{S},n}^{\textsf{S2A}, (r)}|=0$ for all $n\in \mathcal{A}$.    \\
	\STATE \textbf{Output: } Optimal data allocations $\{|D_{\textsf{S},n}^{\textsf{S2A}, (r)}|\}_{n\in\mathcal{A}}$ and $\{|D_{n,k}^{\textsf{A2G}, (r)}|\}_{k\in\mathcal{G}_n}$ for all $n\in\mathcal{A}$. \\
			 \WHILE{$\nu_{U, 1}-\nu_{L,1}\geq \epsilon_1$}
			 		\STATE{Set $X= (\nu_{U, 1}+ \nu_{L,1})/2$}
			 \STATE{Obtain $\{|D_{\textsf{S},n}^{\textsf{S2A}, (r)}|\}_{n\in \mathcal{A}}$ based on the following while loop:}
		\WHILE{$\sum_{n\in\mathcal{A}}|D_{\textsf{S},n}^{\textsf{S2A}, (r)}|<(1-\epsilon_2)X$ or $\sum_{n\in\mathcal{A}}|D_{\textsf{S},n}^{\textsf{S2A}, (r)}|> (1+\epsilon_2)X$ }
		\FOR{each $n\in \mathcal{A}$}
		\STATE{Compute $|D_{\textsf{S},n}^{\textsf{S2A}, (r)}|$  to make $\bar{\tau}_{\textsf{A},n}^{(r)}+\tau_{n,\textsf{S}}^{\textsf{A2S},(r)}$ and  $\frac{1}{2}(\nu_{U, 2}+\nu_{L,2})$ as close as possible within  range $|D_{\textsf{S},n}^{\textsf{S2A},(r)}|\in [0, \min\{|D_{\textsf{S}}^{(r)}|, X\}]$, using bisection search and $\{|D_{n,k}^{\textsf{A2G}, (r)}|\}_{k\in\mathcal{G}_n}$ obtained from  \textbf{Algorithm \ref{algo:Air_Ground}}.}
		\ENDFOR	
		\IF{$\sum_{n\in \mathcal{A}}|D_{\textsf{S},n}^{\textsf{S2A}, (r)}|\leq (1-\epsilon_2)X$}
		\STATE{$\nu_{L, 2} \leftarrow \frac{1}{2}(\nu_{U, 2}+\nu_{L,2})$}
		\ELSE
		\STATE{$\nu_{U, 2} \leftarrow \frac{1}{2}(\nu_{U, 2}+\nu_{L,2})$}
		\ENDIF
		\ENDWHILE
		 \STATE{Compute $\bar{\tau}_{\textsf{A},n}^{(r)}$ in (\ref{eq_sdf})     based on the obtained $\{|D_{n,k}^{\textsf{A2G}, (r)}|\}_{k\in\mathcal{G}_n}$ for all $n\in\mathcal{A}$ and $\{|D_{\textsf{S},n}^{\textsf{S2A}, (r)}|\}_{n\in \mathcal{A}}$.} 
    \STATE{Compute $\bar{\tau}_{\textsf{S}}^{(r)}$ according to (\ref{eq:local_satellite_time}) and $|D_{\textsf{S}}^{(r+1)}|$ in (\ref{eq:sate_dataset_updated})}
				 \STATE{\textbf{if} $\bar{\tau}_{\textsf{S}}^{(r)} \geq \max_{n\in \mathcal{A}}\{\bar{\tau}_{\textsf{A},n}^{(r)}+\tau_{n,\textsf{S}}^{\textsf{A2S},(r)}\}$,  set $\nu_{L,1} = X$.}
				 		\STATE{\textbf{else} set $\nu_{U,1} = X$.}
	 			\ENDWHILE
	\end{algorithmic}
\end{algorithm}

\setlength{\textfloatsep}{12pt}

\textbf{Load balancing across space-air-ground layers:}  Now we revisit our final goal, which is to jointly optimize    $\{|D_{\textsf{S},n}^{\textsf{S2A}, (r)}|\}_{n\in\mathcal{A}}$ and $\{|D_{n,k}^{\textsf{A2G}, (r)}|\}_{k\in\mathcal{G}_n}$ for all $n\in\mathcal{A}$, to minimize the overall  latency $\bar{\tau}^{(r)}$ in (\ref{Lat_with_lllof}) based on the obtained $\bar{\tau}_{S}^{(r)}$, $\bar{\tau}_{\textsf{A},n}^{\textsf{local},(r)}$, and $\bar{\tau}_{\textsf{G},k}^{\textsf{local}}$.  The overall optimization  procedure is summarized in Algorithm \ref{algo:overall}.  Specifically, we solve
\begin{equation}
 \min_{\{|D_{\textsf{S},n}^{\textsf{S2A}, (r)}|\}_{n\in \mathcal{A}},
 }
 \max\Big\{\bar{\tau}_{S}^{(r)}, \max_{n\in \mathcal{A}}\{\bar{\tau}_{\textsf{A},n}^{(r)}+\tau_{n,\textsf{S}}^{\textsf{A2S},(r)}\}  \Big\},
\end{equation}
for load balancing between space and air layers. During optimization, Algorithm  \ref{algo:Air_Ground} is adopted to obtain   $\{|D_{n,k}^{\textsf{A2G}, (r)}|\}_{k\in\mathcal{G}_n}$ for load balancing between air and ground layers and to compute $\bar{\tau}_{\textsf{A},n}^{(r)}$, for a given  $|D_{\textsf{S},n}^{\textsf{S2A}, (r)}|$. Overall, we make  $\bar{\tau}_{S}^{(r)}$ and $ \max_{n\in \mathcal{A}}\{\bar{\tau}_{\textsf{A},n}^{(r)}+\tau_{n,\textsf{S}}^{\textsf{A2S},(r)}\}$ as close as possible by applying bisection search in a hierarchical way.

\begin{remark}
 In practice, Algorithm 1 and Algorithm 2 can be implemented at the nearest gateway to obtain optimized data offloading solutions. The solutions are subsequently sent to the corresponding nodes to execute the data offloading process. 
\end{remark}

\subsection{Case II:    Data Offloading  From   Air/Ground    to Space}
Now we consider Case II, where data samples are transmitted   from air/ground  to  space. 
 Let $|D_\textsf{n,S}^{\textsf{A2S}, (r)}|$ be the number of data samples sent from the air node $n$ to the current satellite.

  \textbf{Dataset  and latency characterization at the space layer:} The satellite-side dataset after data offloading satisfies:
\begin{equation}
|D_{\textsf{S}}^{ (r+1)}| = |D_{\textsf{S}}^{ (r)}| + \sum_{n\in\mathcal{A}}|D_{n, \textsf{S}}^{\textsf{A2S}, (r)}|.
\end{equation}
The satellite-side delay $\bar{\tau}_{\textsf{S}}^{(r)}$  can be  updated accordingly based on  $|D_{\textsf{S}}^{ (r+1)}| = |D_{\textsf{S}}^{ (r)}| + \sum_{n\in\mathcal{A}}|D_{n,\textsf{S}}^{\textsf{A2S}, (r)}|$  and (\ref{eq:local_satellite_time}).

As in Case I, we start by characterizing the data transmission direction between air and ground layers. Without any data transmission between air and ground layers, the completion time at air node $n$ can be written as  follows:
\begin{equation}\label{eqiii}
\tau_{\textsf{A},n}^{\textsf{local},(r)} = \max\Big\{\frac{m_{\textsf{A},n}(|D_{\textsf{A},n}^{(r)}|- |D_{n,\textsf{S}}^{\textsf{A2S}, (r)}|)}{f_{\textsf{A},n}},   \frac{q|D_{n, \textsf{S}}^{\textsf{A2S}, (r)}|}{Z_{n, \textsf{S}}^{\textsf{A2S},(r)}}\Big\},
\end{equation}
when $|D_{n,\textsf{S}}^{\textsf{A2S}(r)}|$ is given. Different from Case I, in (\ref{eqiii}),  both the computation time and the data offloading delay contribute to $\tau_{\textsf{A},n}^{\textsf{local},(r)}$. This is because   the air node can upload the  model to the satellite only when all data samples in $|D_{n, \textsf{S}}^{\textsf{A2S}, (r)}|$ are transmitted to the satellite.

Now we  consider the following two cases, depending on whether air node $n$ should transmit data to the ground layer  or collect data from the ground layer. If $\tau_{\textsf{A},n}^{\textsf{local}, (r)} < \max_{k\in \mathcal{G}_n}\{\tau_{\textsf{G},k}^{\textsf{local}, (r)}+ \tau_{k,n}^{\textsf{G2A}, (r)} \}$,   we let devices in $\mathcal{G}_n$ offload data to  the associated air node $n$ for load balancing.
 Otherwise, we let air node   $n$   transmit data samples to the corresponding ground devices. We consider the first case for description. The result for the second case can be obtained in a similar way.

\textbf{Dataset  and latency characterization at air/ground layers:} Let $D_{k,n}^{\textsf{G2A}, (r)}$ be the dataset that is sent from   ground device  $k\in \mathcal{G}_n$ to air node $n$.  Then, we have
\begin{equation}
|D_{\textsf{A},n}^{ (r+1)}| = |D_{\textsf{A},n}^{ (r)}| -     |D_{n, \textsf{S}}^{\textsf{A2S}, (r)}|  +\sum_{k\in\mathcal{G}_n}|D_{ k,n}^{\textsf{G2A}, (r)}|
\end{equation}
at each air node $n$, after  transmitting $|D_{n,\textsf{S}}^{\textsf{A2S}, (r)}|$  samples to the satellite and   receiving $\sum_{k\in\mathcal{G}_n}|D_{ k,n}^{\textsf{G2A}, (r)}|$  samples from ground devices in cluster $n$. For each ground device $k\in\mathcal{G}_n$, we  obtain
\begin{equation}\label{prk2}
|D_{\textsf{G},k}^{ (r+1)}| = |D_{\textsf{G},k}^{ (r)}|
 -  |D_{k,n}^{\textsf{G2A}, (r)}|
\end{equation}
after transmitting data to the associated air node.

  From these definitions, we obtain the following result:
\begin{align}\label{eq:pwqwp}
&\bar{\tau}_{\textsf{A}, n}^{\textsf{local}, (r)} = \\
&\begin{dcases}  
\max\Big\{\frac{m_{\textsf{A},n}|D_{\textsf{A},n}^{ (r+1)}|}{f_{\textsf{A},n}}, \frac{q|D_{n,\mathcal{S}}^{\textsf{A2S}, (r)}|}{Z_{n,\mathcal{S}}^{\textsf{A2S}, (r)}}\Big\},  \  \  \text{if $|D_{\textsf{A},n}^{(r+1)}| \leq |D_{\textsf{A},n}^{(r)}|$}   \\
\max\Bigg\{\max\Big\{\frac{m_{\textsf{A},n}|D_{\textsf{A},n}^{(r)}|}{f_{\textsf{A},n}},   \max_{k\in\mathcal{G}_n}\{\frac{q|D_{k,n}^{\textsf{G2A}, (r)}|}{Z_{k,n}^{\textsf{G2A},(r)}}\}\Big\} \\ \ \  + \frac{m_{\textsf{A},n}( \sum_{k\in\mathcal{G}_n}|D_{ k,n}^{\textsf{G2A}, (r)}|- |D_{n,\textsf{S}}^{\textsf{A2S}, (r)}|) }{f_{\textsf{A},n}}, \frac{q|D_{n,\mathcal{S}}^{\textsf{A2S}, (r)}|}{Z_{n,\mathcal{S}}^{\textsf{A2S}, (r)}}\Bigg\},\\ \ \ \  \ \text{otherwise}. \nonumber
\end{dcases}
\end{align}
We note that air node $n$ is ready to transmit the model to the satellite when data offloading to satellite is also completed. This is captured in the latency result above.

 At each ground device $k$, we can write
\begin{align} \label{eqdkkklsi}
\bar{\tau}_{\textsf{G},k}^{\textsf{local}, (r)}  =   
\max\Big\{\frac{m_{\textsf{G},k}(|D_{\textsf{G},k}^{(r)}| - |D_{k,n}^{\textsf{G2A}, (r)}|)}{f_{\textsf{G},k}},  \frac{q|D_{k,n}^{\textsf{G2A}, (r)}|}{Z_{k,n}^{\textsf{G2A}, (r)}} 
\Big\},
\end{align}
In (\ref{eqdkkklsi}), we take the maximum of local computation time and data offloading delay. Again, this is because the ground device can start uploading the updated model only if both the local computation and data offloading processes  are completed.

 \textbf{Load balancing between air/ground layers:} For load balancing between air and ground layers, our goal is to optimize $\{|D_{k,n}^{\textsf{G2A}, (r)}|\}_{k\in \mathcal{G}_n}$.  It can be seen that the completion time at the ground layer, i.e., $\max_{k\in\mathcal{G}_n}\{\bar{\tau}_{\textsf{G},k}^{\textsf{local}, (r)} + \tau_{k,n}^{\textsf{G2A}, (r)}\}$,  is a decreasing function of  $|D_{k,n}^{\textsf{G2A}, (r)}|$ if $|D_{k,n}^{\textsf{G2A}, (r)}|\leq \frac{m_{\textsf{G},k}Z_{k,n}^{\textsf{G2A}, (r)}|D_{\textsf{G},k}^{(r)}|}{m_{\textsf{G},k}Z_{k,n}^{\textsf{G2A}, (r)}+qf_{\textsf{G},k}}$, and an increasing function of  $|D_{k,n}^{\textsf{G2A}, (r)}|$ otherwise.
Also, the delay $\bar{\tau}_{\textsf{A},n}^{\textsf{local}, (r)}$ at the air layer is  an increasing function of $|D_{k,n}^{\textsf{G2A}, (r)}|$. Hence, similar to Algorithm \ref{algo:Air_Ground}, we can find \{|$D_{k,n}^{\textsf{G2A}, (r)}|\}_{k\in \mathcal{G}_n}$ by using bisection search in a hierarchical way to make  $\bar{\tau}_{\textsf{A},n}^{(r)}$ and $\max_{k\in\mathcal{G}_n}\{\bar{\tau}_{\textsf{G},k}^{\textsf{local}, (r)} + \tau_{k,n}^{\textsf{G2A}, (r)}\}$ as close as possible within the range following range:
\begin{equation}
|D_{k,n}^{\textsf{G2A}, (r)}|\in\Big[0, \min\Big\{\frac{m_{\textsf{G},k}Z_{k,n}^{\textsf{G2A}, (r)}|D_{\textsf{G},k}^{(r)}|}{m_{\textsf{G},k}Z_{k,n}^{\textsf{G2A}, (r)}+qf_{\textsf{G},k}}, |D_{\textsf{G},k}^{(r)}|-|D_k^l|\Big\}\Big].
\end{equation}
Recall that $|D_k^l|$ is the number of privacy-sensitive samples at ground device $k$. Hence, 
$|D_{\textsf{G},k}^{(r)}|-|D_k^l|$ represents  the amount of non-sensitive data of ground device $k$ at round $r$, which captures the feasible number of samples for offloading.

\textbf{Load balancing across space-air-ground layers:}  Finally, we optimize   $\{|D_{n,\textsf{S}}^{\textsf{A2S}, (r)}|\}_{n\in\mathcal{A}}$ and  $\{|D_{k,n}^{\textsf{G2A}, (r)}|\}_{k\in\mathcal{G}_n}$ for all $n\in\mathcal{A}$, to minimize the overall  latency $\bar{\tau}^{(r)}$ in (\ref{Lat_with_lllof}) based on the obtained $\bar{\tau}_{S}^{(r)}$, $\bar{\tau}_{\textsf{A},n}^{\textsf{local},(r)}$, and $\bar{\tau}_{\textsf{G},k}^{\textsf{local}}$. Similar to Algorithm \ref{algo:overall} for Case I, we can make   $\bar{\tau}_{S}^{(r)}$ and $ \max_{n\in \mathcal{A}}\{\bar{\tau}_{\textsf{A},n}^{(r)}+\tau_{n,\textsf{S}}^{\textsf{A2S},(r)}\}$ as close as possible by applying bisection search, where the solution for load balancing between air and ground layers is adopted during this process.

 \subsection{Complexity Analysis}
Algorithm \ref{algo:Air_Ground} involves load balancing between an air node and the associated ground devices, utilizing nested loops and bisection searches. The primary loop, governed by the variables $\nu_{L,1}$ and $\nu_{L,2}$, iterates using a bisection method until a specified precision $\epsilon_1$ is achieved, contributing a complexity of $\mathcal{O}(\log(\frac{1}{\epsilon_1}))$ \cite{sikorski1982bisection,wang2022interference}. Within this loop, an inner loop also utilizes bisection search to meet a precision $\epsilon_2$, adding a complexity of $\mathcal{O}(\log(\frac{1}{\epsilon_2}))$. The for-loop iterates over $n$ ground devices, with each iteration involving a bisection search that contributes $\mathcal{O}\left(\log \left(\min\{|D_{\textsf{A},n}^{(r)}|, Y_n\}\right)\right)$ complexity \cite{flores1971average,shi2016smoothed}. Summing these, the overall time complexity of Algorithm \ref{algo:Air_Ground} can be written as $\mathcal{O}\left(\log(\frac{1}{\epsilon_1}) \times \log(\frac{1}{\epsilon_2}) \times |\mathcal{G}_n| \times \log \left(\min\{|D_{\textsf{A},n}^{(r)}|, Y_n\}\right) \right)$, reflecting the combined logarithmic and linear components of the nested operations. Similarly,  the complexity of Algorithm \ref{algo:overall} becomes $\mathcal{O}\left(\log(\frac{1}{\epsilon_1}) \times \log(\frac{1}{\epsilon_2}) \times |\mathcal{A}| \times \log \left(\min\{|D_{\textsf{S}}^{(r)}|, X\}\right) \right)$.

\section{Convergence Analysis}
\label{sec:convergence}

In this section, we investigate the convergence property of the proposed algorithm. 
After data offloading is performed in the $r$-th  training round, the global loss function defined  in (\ref{original_loss}) can be rewritten in the following form:
\begin{align}\label{global_FL}
	  F(\mathbf{w})  = & \sum_{k\in \mathcal{G}} \lambda_{\textsf{G}, k}^{(r)} \ell_{\textsf{G}, k}^{(r+1)}(\mathbf{w}) + \sum_{n\in \mathcal{A}} \lambda_{\textsf{A}, n}^{(r)} \ell_{\textsf{A}, n}^{(r+1)}(\mathbf{w})\nonumber \\ &+ \lambda_{\textsf{S}}^{(r)} \ell_{\textsf{S}}^{(r+1)} (\mathbf{w}).
\end{align} 
We note that the global loss function $F(\mathbf{w})$ is time-invariant 
 because the global dataset does not change; rather only the data samples are exchanged among the nodes.  
On the other hand, the local losses, i.e., $\ell_{\textsf{G}, k}^{(r+1)}(\mathbf{w})$, $\ell_{\textsf{A}, n}^{(r+1)}(\mathbf{w}) $, and $\ell_{\textsf{S}}^{(r+1)} (\mathbf{w})$,  vary over time.  
 Our goal is to  analyze the evolution of $\|\nabla F(\mathbf{w}^{(r)})\|$  
 to characterize  the convergence behavior for non-convex loss functions. We rely on the following assumptions.

\begin{assumption}\label{assump_smooth} 
$\ell_{\textsf{G}, k}^{(r+1)}(\mathbf{w})$, $\ell_{\textsf{A}, n}^{(r+1)}(\mathbf{w}) $ and $\ell_{\textsf{S}}^{(r+1)} (\mathbf{w})$, are $L$-smooth for any $k\in \mathcal{G}$, $n\in \mathcal{A}$, and for any $r$.
\end{assumption}
\begin{assumption}\label{bound_variance}
The mini-batch gradients $\tilde{\nabla}\ell_{\textsf{G}, k}^{(r+1)}(\mathbf{w})$, $\tilde{\nabla}\ell_{\textsf{A}, n}^{(r+1)}(\mathbf{w})$, and $\tilde{\nabla}\ell_{\textsf{S}}^{(r+1)}(\mathbf{w})$  are unbiased estimates of ${\nabla}\ell_{\textsf{G}, k}^{(r+1)}(\mathbf{w})$, ${\nabla}\ell_{\textsf{A}, n}^{(r+1)}(\mathbf{w})$, and ${\nabla}\ell_{\textsf{S}}^{(r+1)}(\mathbf{w})$, respectively. The  variance  is bounded as $\mathbb{E} \|\tilde{\nabla}\ell_{\textsf{G}, k}^{(r+1)}(\mathbf{w}) - {\nabla}\ell_{\textsf{G}, k}^{(r+1)}(\mathbf{w})\|^2 \leq \sigma_g^2$, $\forall k\in \mathcal{G}$, which also holds for $\tilde{\nabla}\ell_{\textsf{A}, n}^{(r+1)}(\mathbf{w})$ and $\tilde{\nabla}\ell_{\textsf{S}}^{(r+1)}(\mathbf{w})$.   	
\end{assumption}
\begin{assumption}
\label{bound_heterogeneity}
The gradient dissimilarity between each local loss function 
 and the global loss function   $F(\mathbf{w})$  is bounded as 
$\mathbb{E} \left \|  \nabla \ell_{\textsf{G}, k}^{(r+1)}(\mathbf{w})   \! -\! F(\mathbf{w}) \right\|^2 \leq c_r\left \|F(\mathbf{w}) \right\|^2+  \delta_r^2$,  
$\forall k\in \mathcal{G}$. This holds for $\ell_{\textsf{A}, n}^{(r+1)}(\mathbf{w})$, $\forall n\in \mathcal{A}$ and $\ell_{\textsf{S}}^{(r+1)}(\mathbf{w})$ as well.
\end{assumption}
Assumptions \ref{assump_smooth}--\ref{bound_heterogeneity} are standard and have been widely adopted in the analyses of existing works    \cite{wang2019matcha, hosseinalipour2023parallel,ganguly2023multi}, where Assumption \ref{bound_heterogeneity} specifically quantifies the data heterogeneity in each   round  $r$.  
We present our main theorem below.

\begin{theorem} \label{theorem_full}
	Suppose that Assumptions \ref{assump_smooth}--\ref{bound_heterogeneity} hold and the learning rates satisfies
	\equa{\label{theorem_eta_full}
\eta_{\textsf{G}, k}^{(r)} = \eta_{\textsf{A}, n}^{(r)} = \eta_{\textsf{S}}^{(r)} = \eta^{(r)} \leq \frac{1}{2\sqrt{1\!+\!c_r}HL},
} 
where $H$ denotes the number of local iterations at each node per global round. Then  under non-convex settings, our algorithm  satisfies the following convergence result:
\begin{align}\label{equa_theorem}
 & \frac{1}{\Gamma_R} \sum_{r=0}^{R-1} \eta^{(r)} \mathbb{E} \Big\|\nabla F\Big(\mathbf{w}^{(r)}\Big)\Big\|^2
 \leq 4\frac{F\left(\mathbf{w}^{(0)}\right)  - F^*}{ H\Gamma_R} \nonumber \\
 & +  \frac{4L}{\Gamma_R} \sum_{r=0}^{R-1}  (\eta^{(r)})^2 \!\Big( \sum_{k\in \mathcal{G}}  (\lambda_{\textup{\textsf{G}}, k}^{(r)})^2  +  \sum_{n\in \mathcal{A}} (\lambda_{\textup{\textsf{A}}, n}^{(r)})^2 + (\lambda_{\textup{\textsf{S}}}^{(r)})^2 \Big)   \sigma_g^2 \nonumber \\
 & +  \frac{2H^2 L^2 \sigma_g^2}{\Gamma_R}\sum_{r=0}^{R-1}  (\eta^{(r)})^3  + \frac{4H^2 L^2}{\Gamma_R} \sum_{r=0}^{R-1}  (\eta^{(r)})^3 \delta_r^2,
\end{align}
where  $F^*$ is the minimum value that $F(\mathbf{w})$ can achieve and $\Gamma_R = \sum_{r=0}^{R-1}  \eta^{(r)}$ is the summation of learning rates.
\end{theorem}
\begin{proof}
    See Appendix \ref{proof_theorem}
\end{proof}
The impact of data heterogeneity after each round of data offloading  is reflected both in the  learning rate condition (\ref{theorem_eta_full}) and the last term of (\ref{equa_theorem}) in the convergence bound. From (\ref{theorem_eta_full}), we see that as the extent of data heterogeneity after data offloading gets larger,   a smaller learning rate is required to guarantee the convergence of the  algorithm. We also observe from (\ref{equa_theorem}) that the bound increases as the heterogeneity of data distributions across the nodes grows. The second term of the right-hand side of (\ref{equa_theorem}) captures the effect of the portion of data samples at each node on the convergence bound, which is time-varying due to data offloading. Additionally,    by selecting an appropriate learning rate that satisfies $\sum_{r=0}^{R-1}  (\eta^{(r)})^2 \rightarrow 0$,    $\sum_{r=0}^{R-1}  (\eta^{(r)})^3 \rightarrow 0$ and $\Gamma_R\rightarrow \infty$ for $R\rightarrow \infty$, the upper bound  
will diminish to zero. In particular, we can either adopt a decaying learning rate according to $\eta^{(r)} = \frac{\eta^{(0)}}{r+1}$ or keep it constant as $\eta^{(r)} = \frac{1}{\sqrt{HR}}$. This guarantees convergence to a stationary point of the non-convex loss function.

\section{Experimental Results} \label{sec:experiments}
In this section, we provide experimental results to validate the effectiveness of the proposed   methodology in SAGINs.

 \subsection{Simulation Setup}

\textbf{Dataset and model:} We consider the following benchmark datasets for FL: MNIST, FMNIST, and CIFAR-10. Using MNIST and FMNIST, we train a convolutional neural network with two convolutional layers and two fully connected layers, and a convolutional neural network with two convolutional layers and one fully connected layer,   respectively. Using CIFAR-10, we train the VGG-11 model.  We conduct FL using the  training set of each dataset, and evaluate the performance of the constructed global model using the  testing set.

\textbf{SAGIN setting:} We consider  $K=50$ ground devices located at a squared target region of $1200$ m $\times$ $1200$ m. There are $N=5$ air nodes at a height of $20$ km above the target area, each serving $10$ ground devices without overlapping. A series of LEO satellites cover the target region in each global round, where we adopt the \texttt{walkerStar} function in MATLAB to construct a constellation model.  Fig. \ref{fig:constellation} shows the created satellite constellation, where   $80$  LEO  satellites are distributed evenly across $5$ different orbits with altitude of $800$ km and inclination of $85^{\circ}$. We set the minimum elevation angle to communicate to $15^{\circ}$, and  the    latitude  and longitude of  the target region are $40^{\circ}$ N and $86^{\circ}$ W, respectively. We use \texttt{accessIntervals} function to calculate the coverage time of each satellite over the target region.  Referring to the   settings of prior works \cite{rodrigues2023hybrid, fang2023olive, Razmi1}, we  adopt the  following parameter values for simulations: $f_{\textsf{G},k}=10^8$ Hz, $f_{\textsf{A},n}=10^9$ Hz, $f_{\textsf{S},i}\in [1, 10]\times10^9$ Hz,  $m_{\textsf{G},k}=m_{\textsf{A},n}=m_{\textsf{S}} = 3\times10^9$ cycles/sample, $p_{\textsf{G},k}=0.1$ W, $p_{\textsf{A},n}=1$ W, $p_{\textsf{S},i}=10$ W, $Z_{i, i+1}^{\textsf{ISL},(r)}=3.125$ Mbps, $N_0=3.98\times10^{-21}$ W/Hz. Here, to model the time-varying resource availability at the space layer over the target region, the CPU frequencies of satellites $f_{\textsf{S},i}$ are sampled from  a specific uniform distribution $[1, 10]\times10^9$.

The training set of each dataset is distributed to the ground devices in two different scenarios: IID (independent and identically distributed) and non-IID cases. For the IID case, we allocate the training samples to the ground devices  uniformly  at random. For the non-IID scenario, we sort the training set according to each sample's class, split the sorted dataset into 200 shards, and then randomly assign 4 shards to each ground device. This introduces heterogeneous data distributions among ground devices. Note that the nodes in the space and air layers do not hold data at the beginning. We set the portion of non-sensitive   to $\alpha_k=\alpha=0.8$ for all ground devices, and also study the  effect  of $\alpha$ in Section \ref {subsec:ablation}.

\begin{figure}[t]
  \centering
         \includegraphics[width=0.36\textwidth]{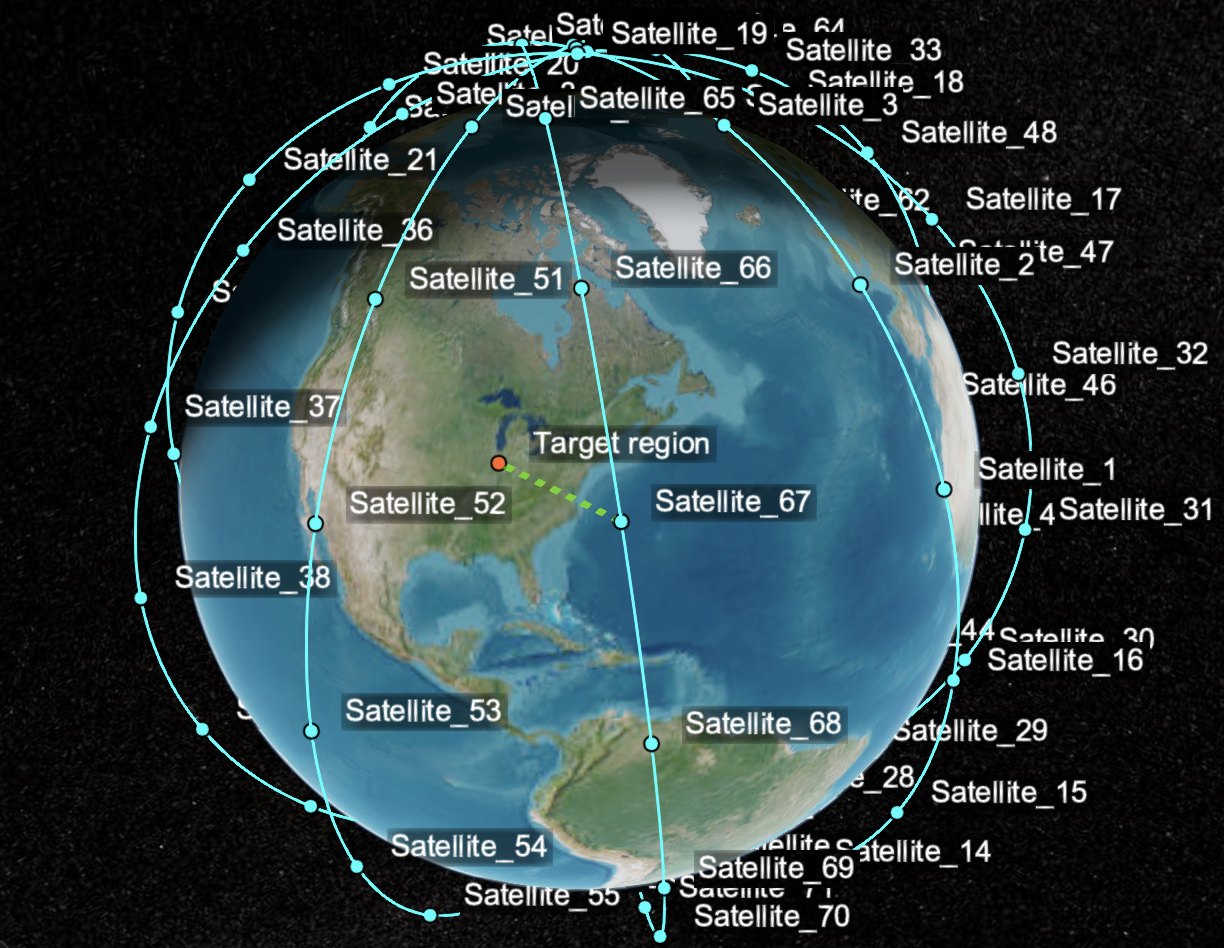} 
      \caption{Illustration of the satellite constellation   constructed based on the \texttt{walkerStar} function.}\label{fig:constellation}
      \vspace{-1.5mm}
   \end{figure}
\begin{figure*}[t]
        \vspace{-0.5mm}
    \centering
 \begin{subfigure}[b]{0.3\textwidth}
         \centering
    \includegraphics[width=\textwidth]{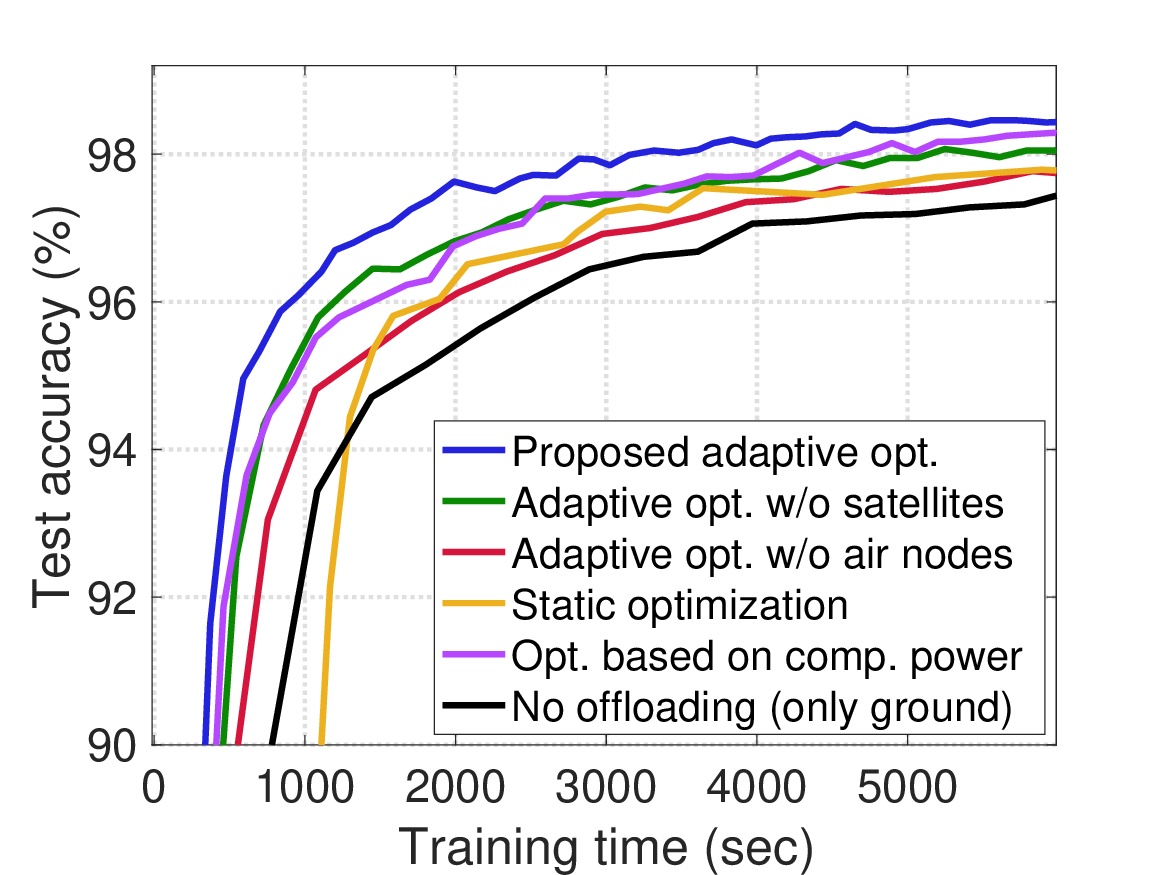}
            \vspace{-4.5mm}
       \caption{MNIST, IID} 
        \vspace{-0.5mm}
 \end{subfigure}  
    \begin{subfigure}[b]{0.3\textwidth}
         \centering
  \includegraphics[width=\textwidth]{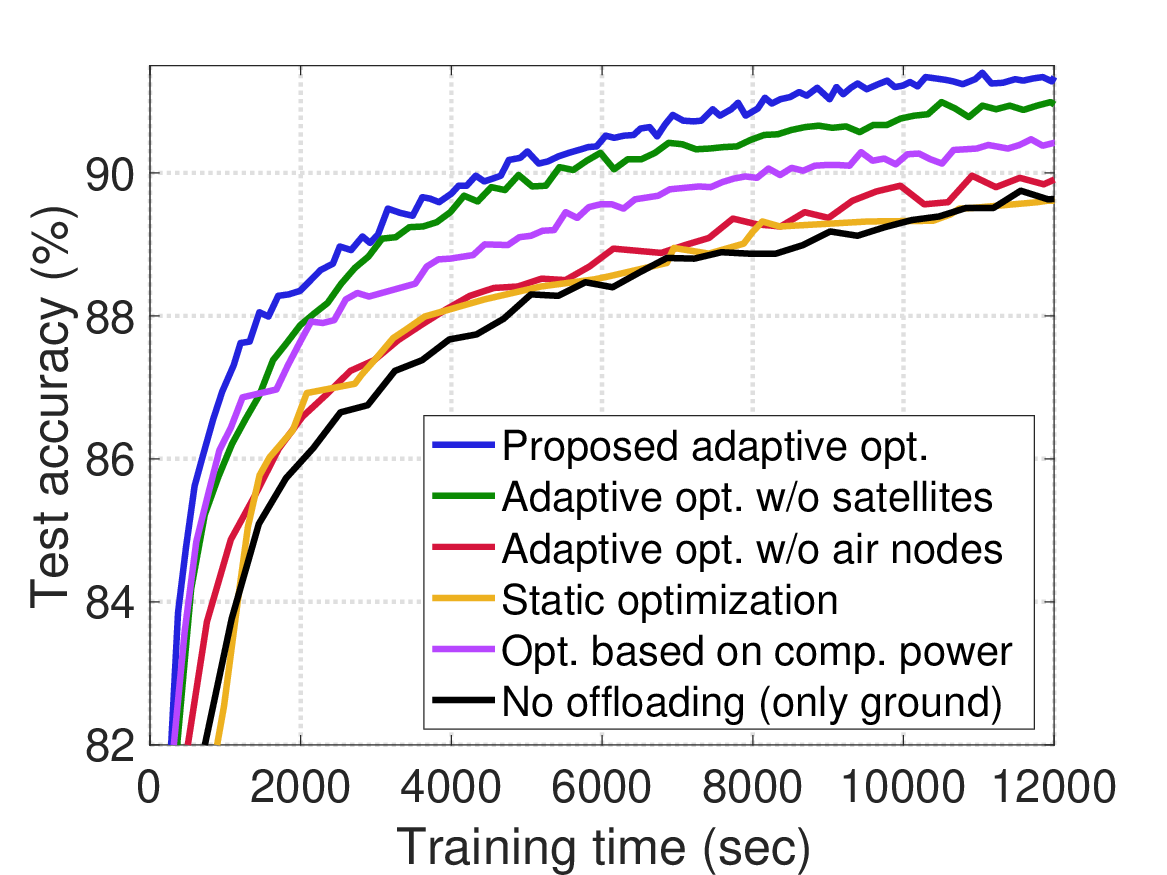}
              \vspace{-4.5mm}
       \caption{FMNIST, IID}
               \vspace{-0.5mm}
  \end{subfigure}
 \begin{subfigure}[b]{0.3\textwidth}
         \centering
    \includegraphics[width=\textwidth]
 {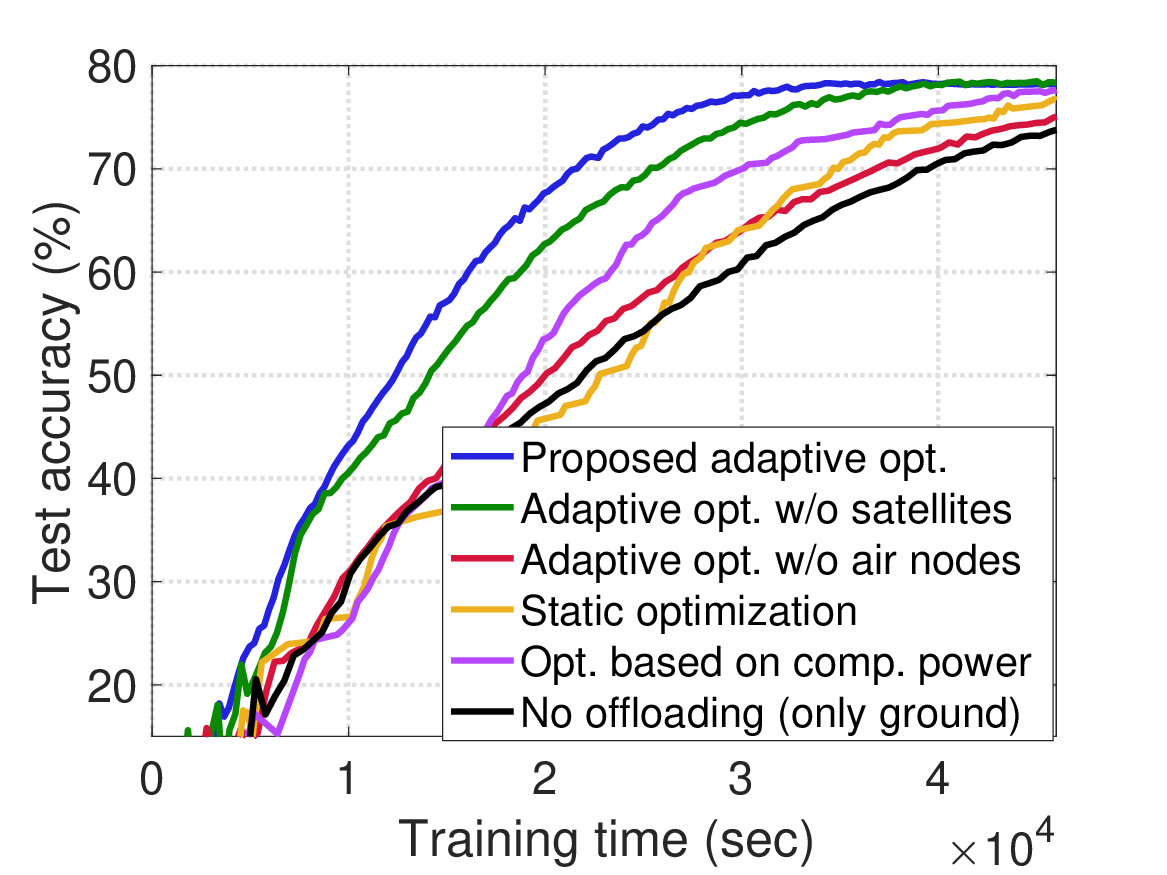}
             \vspace{-4.5mm}
       \caption{CIFAR-10, IID} 
               \vspace{-0.5mm}
  \end{subfigure}
       \begin{subfigure}[b]{0.3\textwidth}
         \centering
\includegraphics[width=\textwidth] {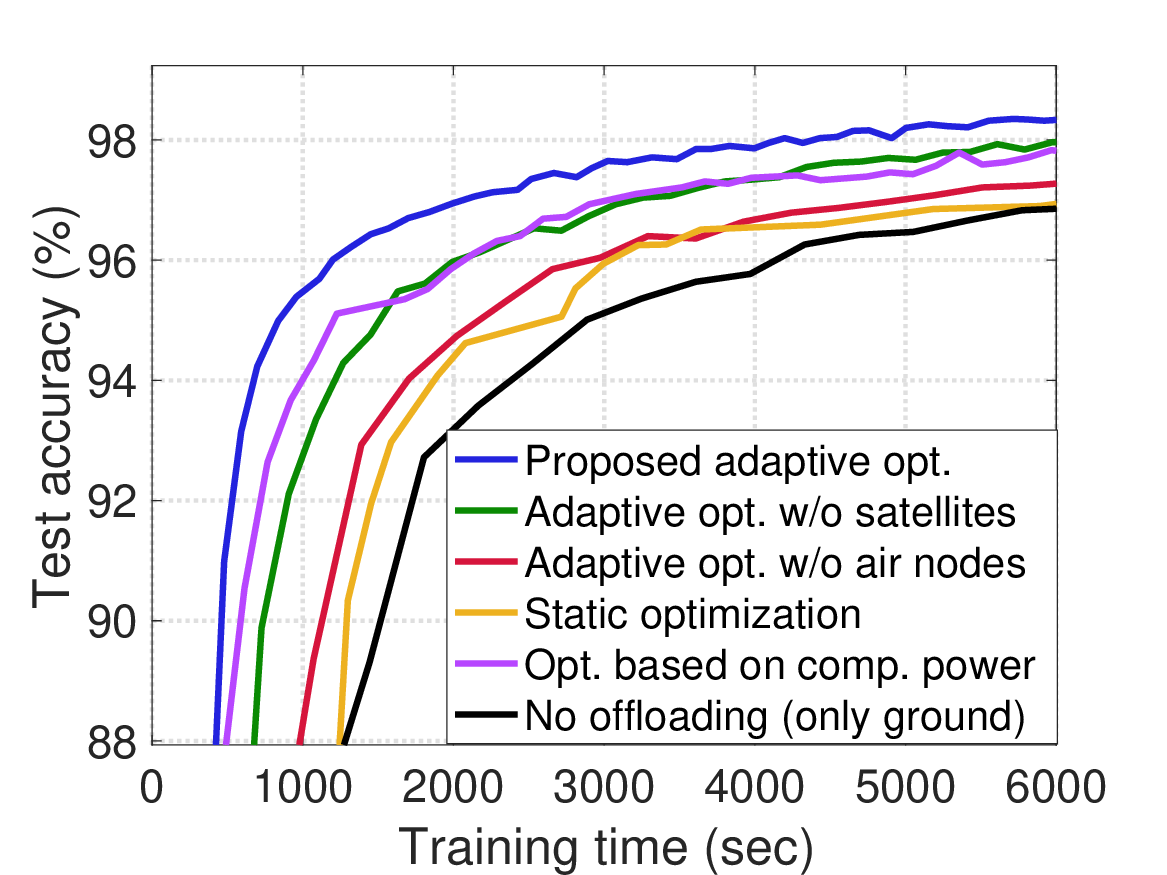}
            \vspace{-4.5mm}
    \caption{MNIST, Non-IID}
 \end{subfigure}  
    \begin{subfigure}[b]{0.3\textwidth}
         \centering
\includegraphics[width=\textwidth]
{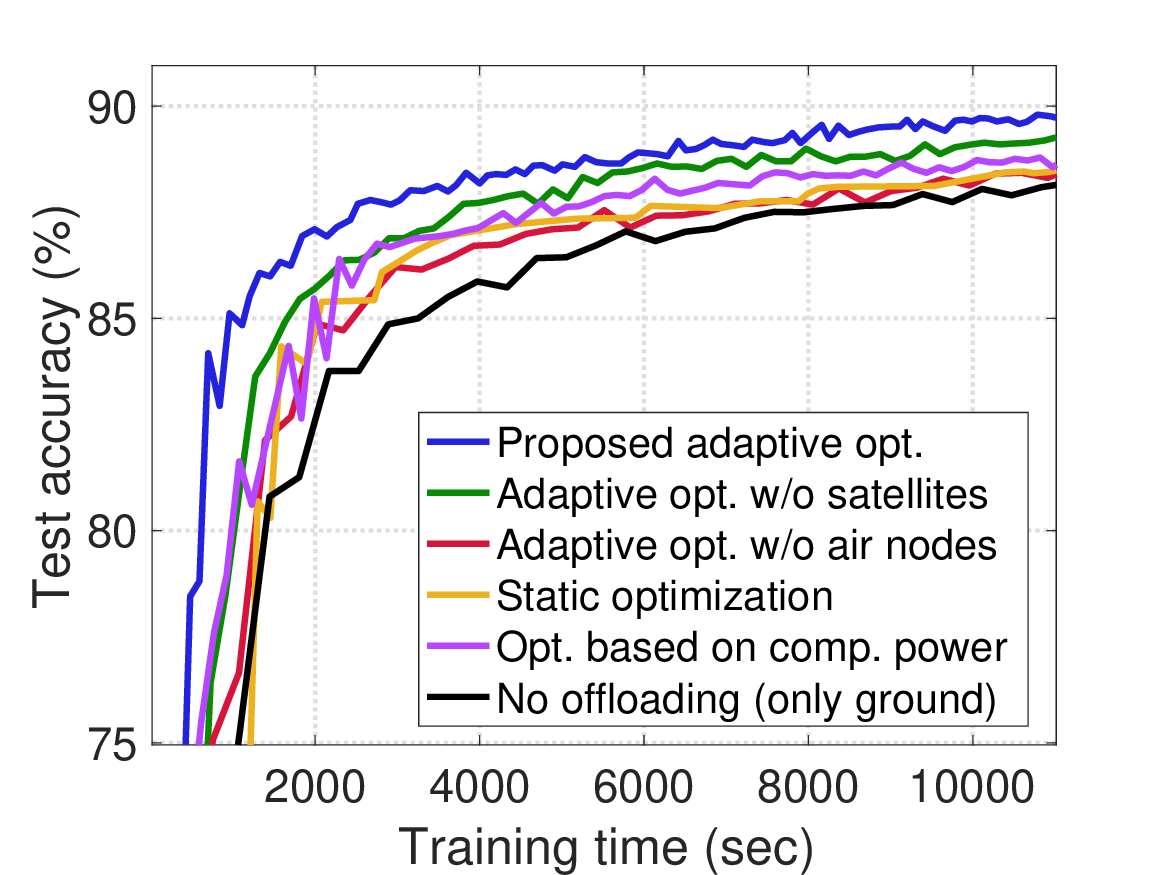}
            \vspace{-4.5mm}
    \caption{FMNIST, Non-IID}
  \end{subfigure}
     \begin{subfigure}[b]{0.3\textwidth}
         \centering
    \includegraphics[width=\textwidth]{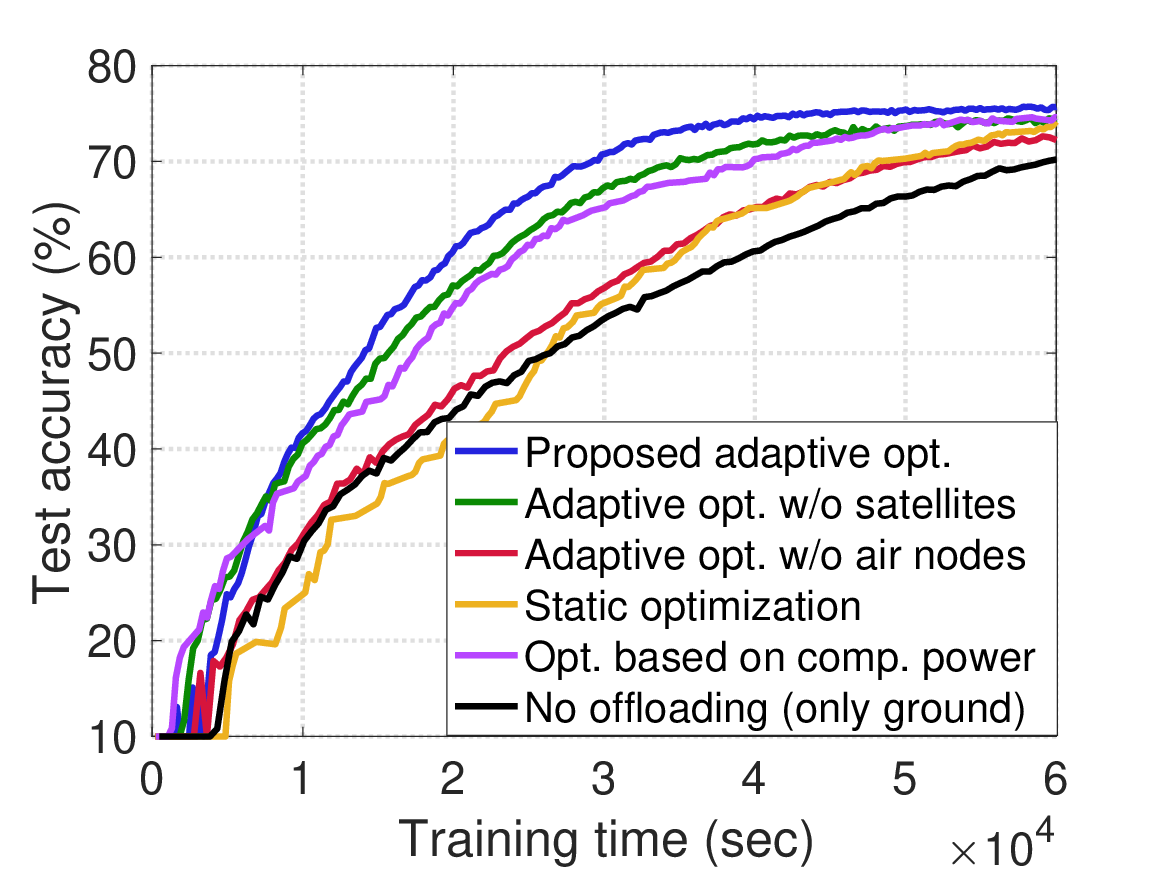}
                \vspace{-4.5mm}
       \caption{CIFAR-10, Non-IID}
  \end{subfigure}
  \vspace{-0.5mm}
     \caption{Accuracy versus training time plots. For the static optimization scheme, we apply our inter-layer data offloading scheme only in the first global round and keep the intra-layer data fixed throughout the remaining rounds.  The results show  the advantage of adaptive data offloading optimization considering both space and air layers. }
    \label{fig:exp_overall}
     \end{figure*}

\textbf{Comparison schemes:}  For baselines,  we first consider  the scheme where only the ground devices process data without any data offloading, to see the advantage of adopting  nodes in space and air layers as edge computing units.  Satellites and air nodes are only used to aggregate the updated models. This baseline represents the majority of existing works that do not involve data offloading. Secondly, we consider optimizing data offloading only between the air and ground layers. Hence, the   satellite-side  computation power is not utilized during local model updates. Similarly, we optimize data offloading only between ground and space layers, without using the computational capabilities of air nodes during local update.  We also consider the static optimization scheme, which applies our  optimization strategy   only at the initial global round and keeps the same solution throughout the remaining FL process. This baseline utilizes the computational resources of all three layers of SAGINs and is considered to see the impact of adaptive data offloading instead of using a fixed solution. Finally, we consider another baseline that utilizes the resources of all layers of SAGINs, where the number
of data samples processed at each node is proportional to its computational power.  For a fair comparison, we use FedAvg   to aggregate the models in all baseline schemes  and our methodology.

\subsection{Main Experimental Results}\label{subsec:exp}

We first observe  Fig. \ref{fig:exp_overall}, which reports the accuracy versus training time plots in different settings. Our key takeaways are as follows. First, the scheme   without  data offloading  achieves slow convergence, since the computation resources of space and air nodes are not utilized in this method.  Utilizing only the computation resources of ground devices causes delays.    We also observe that the fixed data offloading scheme achieves relatively low performance since the varying resource availability at the satellites are not considered in the scheme. If too many or too    few data samples are offloaded to the space layer,  the training process can be slowed down. This highlights the importance of adaptively optimizing   data offloading,  instead of   relying on a fixed solution.  We see that our approach, which leverages both the space and air layers,  attains superior performance compared to the baselines that utilize  only one of these layers.   The proposed scheme also outperforms the scheme with optimized fixed data offloading and the baseline that conducts data offloading proportional to the computational power of each node in SAGINs.   Further ablation studies on  the effect of each layer are provided in the next  subsection. The overall results highlight the significance of (i) inter-layer data offloading across space-air-ground, and (ii) adaptively conducting this to account for the network dynamics in SAGINs.
 
 \begin{figure}[t]
    \centering
 \begin{subfigure}[b]{0.48\textwidth}
         \centering
    \includegraphics[width=\textwidth]{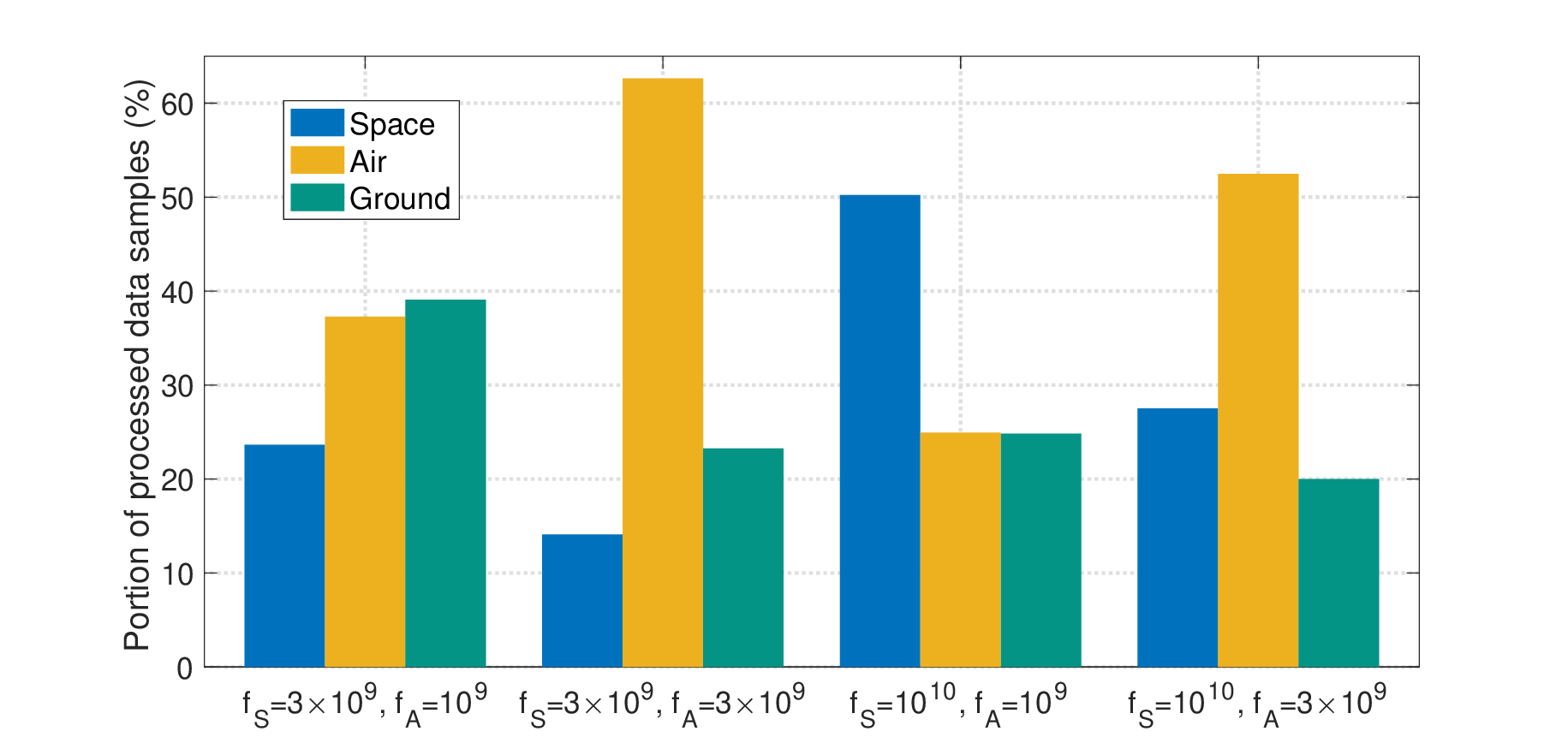}
    \vspace{-6mm}
       \caption{Portion of data samples processed at each layer. We increase the CPU frequency of the (i) air node, (ii) space node, and (iii) both. } \label{subfig_porton_smaples}
 \end{subfigure}  
       \begin{subfigure}[b]{0.24\textwidth}
         \centering
\includegraphics[width=\textwidth] {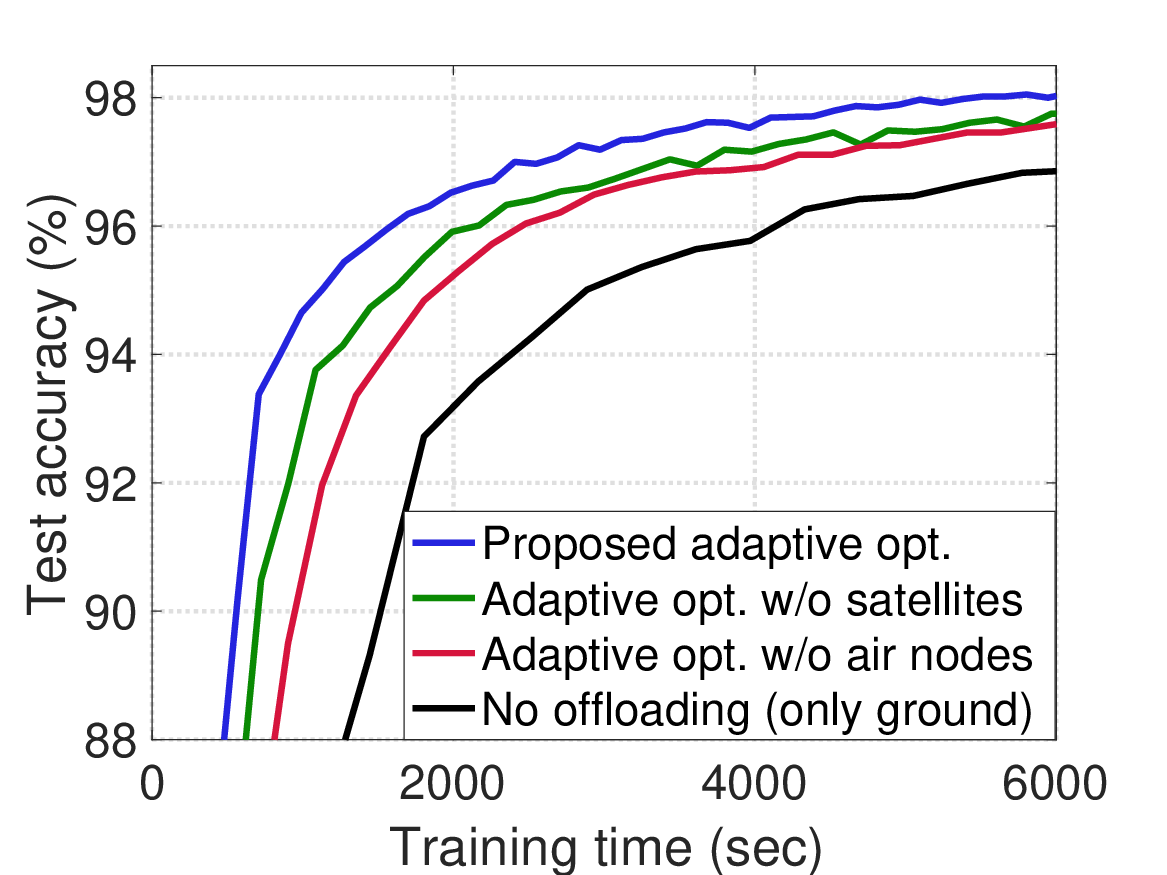}
\caption{Performance with $f_{\textsf{S}}=3\times 10^{9}$ Hz, $f_{\textsf{A}}=10^{9}$ Hz.} \label{dkdohg}
 \end{subfigure}  
    \begin{subfigure}[b]{0.24\textwidth}
         \centering
\includegraphics[width=\textwidth]
{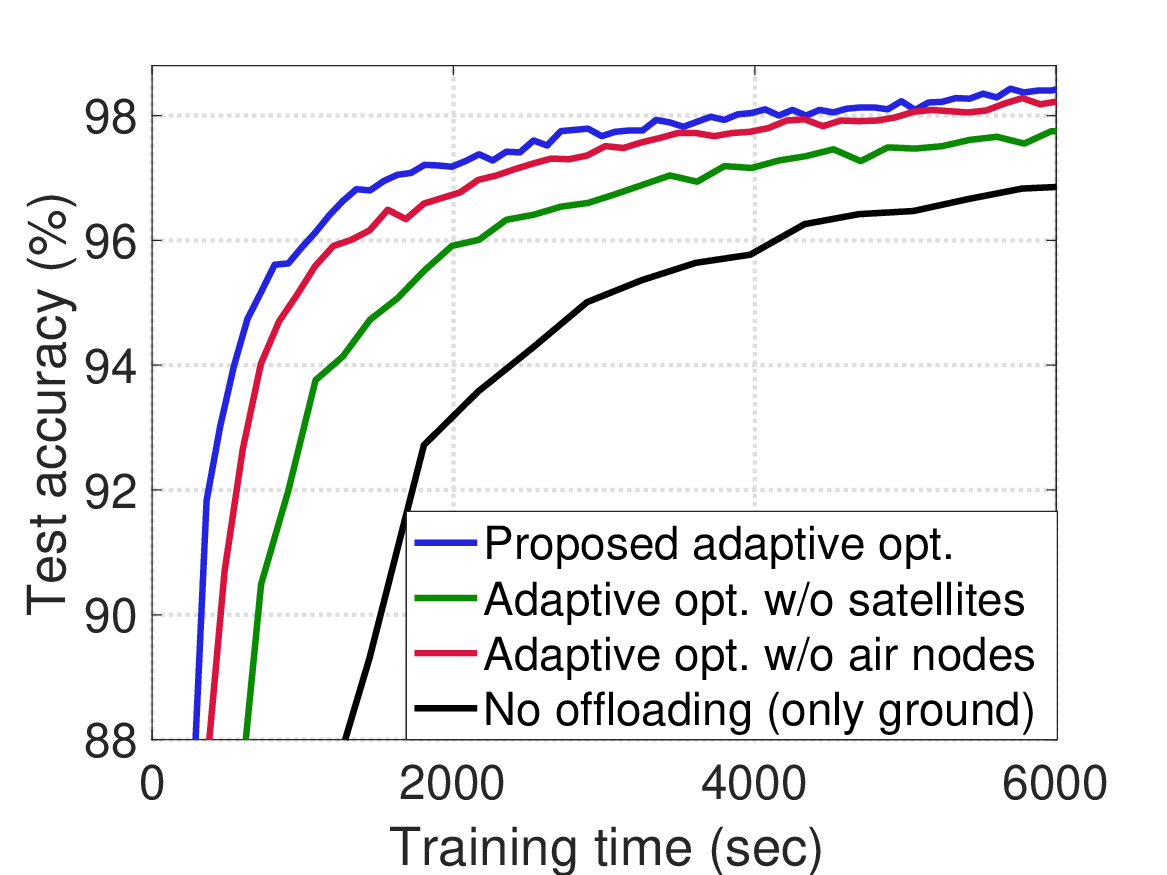}
\caption{Performance with $f_{\textsf{S}}=  10^{10}$ Hz, $f_{\textsf{A}}=10^{9}$ Hz.}\label{subfig___3}
  \end{subfigure}
     \caption{Effect of computation capabilities of space/air nodes. }
\label{fig:satellite_power}
     \end{figure}

\subsection{Varying System Parameters}\label{subsec:ablation}
 
\textbf{Effect of  computation powers of space and air nodes:}
 In Fig. \ref{fig:satellite_power}, we investigate   the effect of computational capabilities at different layers, which can be adapted based on the battery constraint of each node.  In extreme cases, the CPU frequency can drop to 0 if the battery is close to
0, and it can reach the maximum CPU frequency if the battery is sufficient.  MNIST is considered in a non-IID setup. For these experiments, we set the CPU frequencies of space and air nodes (i.e., $f_{\textsf{S}}$ and $f_{\textsf{A}}$, respectively) to the values depicted in the figure. Fig. \ref{subfig_porton_smaples} first shows the portion of data samples processed at each layer in our solution, depending on $f_{\textsf{S}}$ and $f_{\textsf{A}}$. In the first case with $f_{\textsf{S}}=3\times 10^{9}$ Hz and $f_{\textsf{A}}=10^{9}$ Hz (a scenario where both space and air nodes have insufficient battery), a relatively large number of data samples are allocated to the ground layer due to the limited batteries at the space and air nodes.  The air layer is allocated with more data samples than the space layer, indicating that the air nodes are considered more important than the satellites. This can be  also confirmed from  the accuracy curve in Fig. \ref{dkdohg}, by comparing the scheme  without satellites and  the one without air nodes. Now if   $f_{\textsf{A}}$ increases from $10^{9}$ Hz to $3\times 10^{9}$ Hz (i.e., a scenario where the air node has more battery compared to the previous case), the portion of data samples processed at the air node becomes more dominant.  On the other hand, if we increase $f_{\textsf{S}}$ from $10^{9}$ Hz to $10^{10}$ Hz (i.e., if the satellite has sufficient battery) while setting $f_{\textsf{A}}=10^{9}$ Hz, the role of the space layer becomes crucial,   as   also verified in Fig.  \ref{subfig___3}. Finally, when both space and air layers have sufficient resources ($f_{\textsf{S}}=  10^{10}$ Hz and $f_{\textsf{A}}=3\times10^{9}$ Hz), only $20\%$ of data samples are allocated to the ground layer. This allocation  is the minimum amount of data that should be processed at the ground layer considering the portion of non-sensitive samples ($\alpha=0.8$). Again, the   results underscore the significance of  taking advantage of the computation resources across all layers in SAGINs during the FL process.

\begin{figure}[t]
    \centering

       \begin{subfigure}[b]{0.24\textwidth}
         \centering
\includegraphics[width=\textwidth]{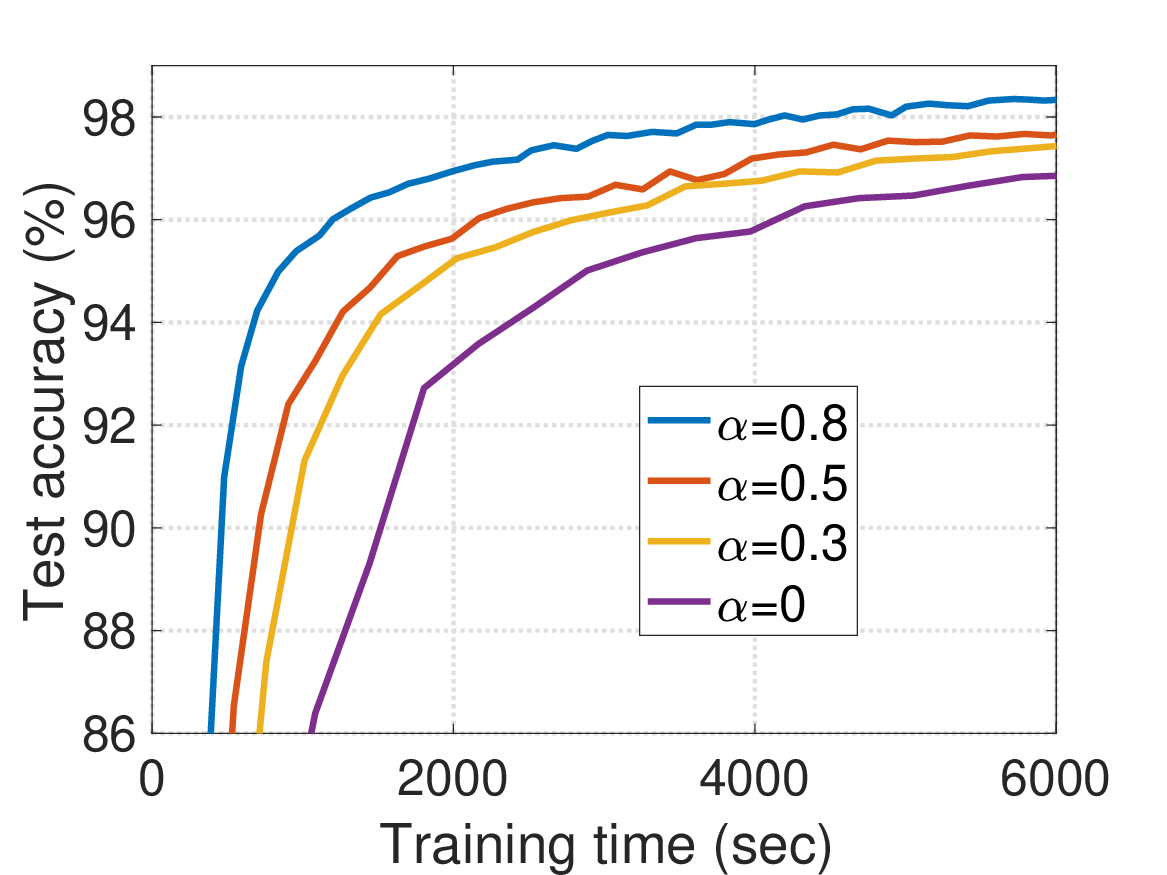}
    \caption{Accuracy versus training time on MNIST.} 
 \end{subfigure}  
    \begin{subfigure}[b]{0.24\textwidth}
         \centering
\includegraphics[width=\textwidth]
{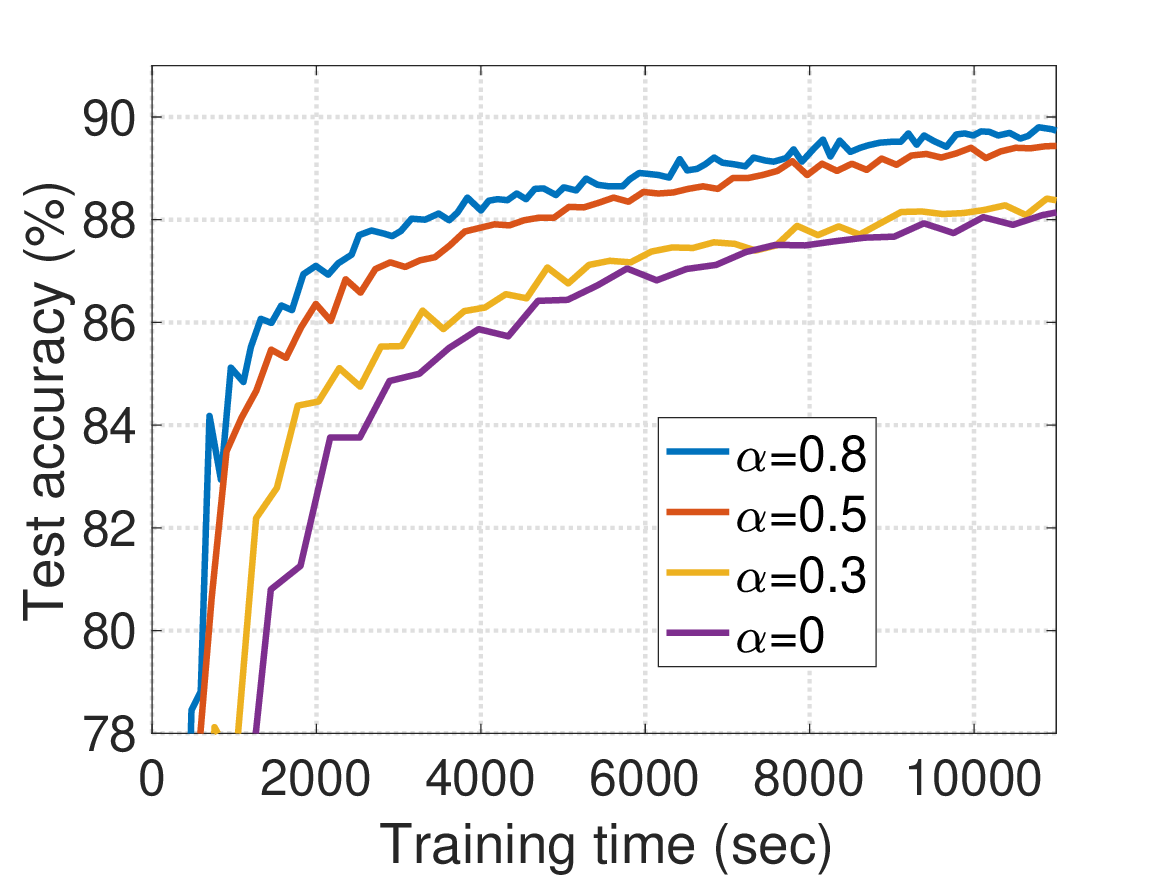}
    \caption{Accuracy versus training time on FMNIST.} 
  \end{subfigure}
       \begin{subfigure}[b]{0.24\textwidth}
         \centering
\includegraphics[width=\textwidth]{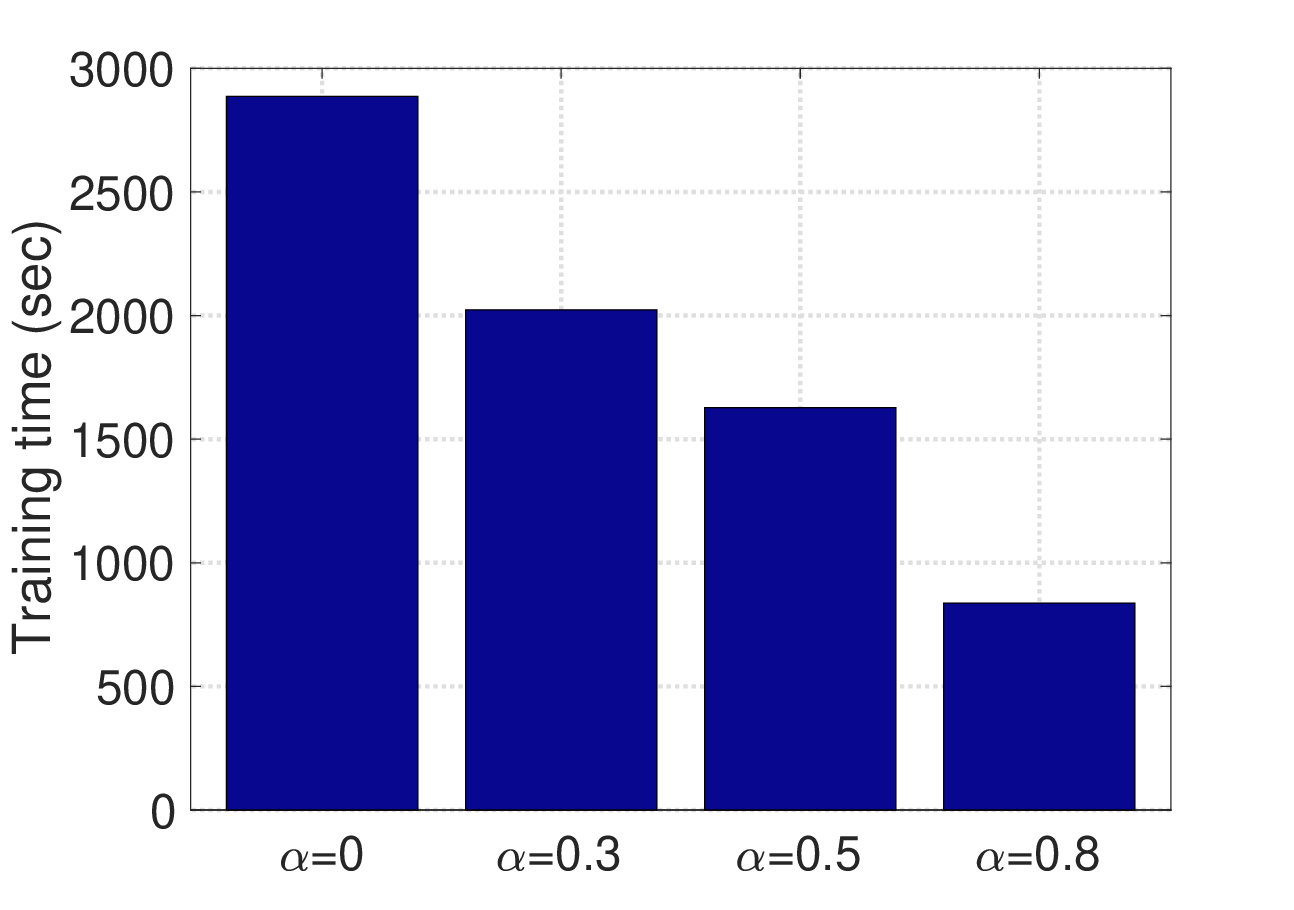}
    \caption{Training time to achieve 95\% accuracy on MNIST.} 
 \end{subfigure} 
   \begin{subfigure}[b]{0.24\textwidth}
         \centering
\includegraphics[width=\textwidth]
{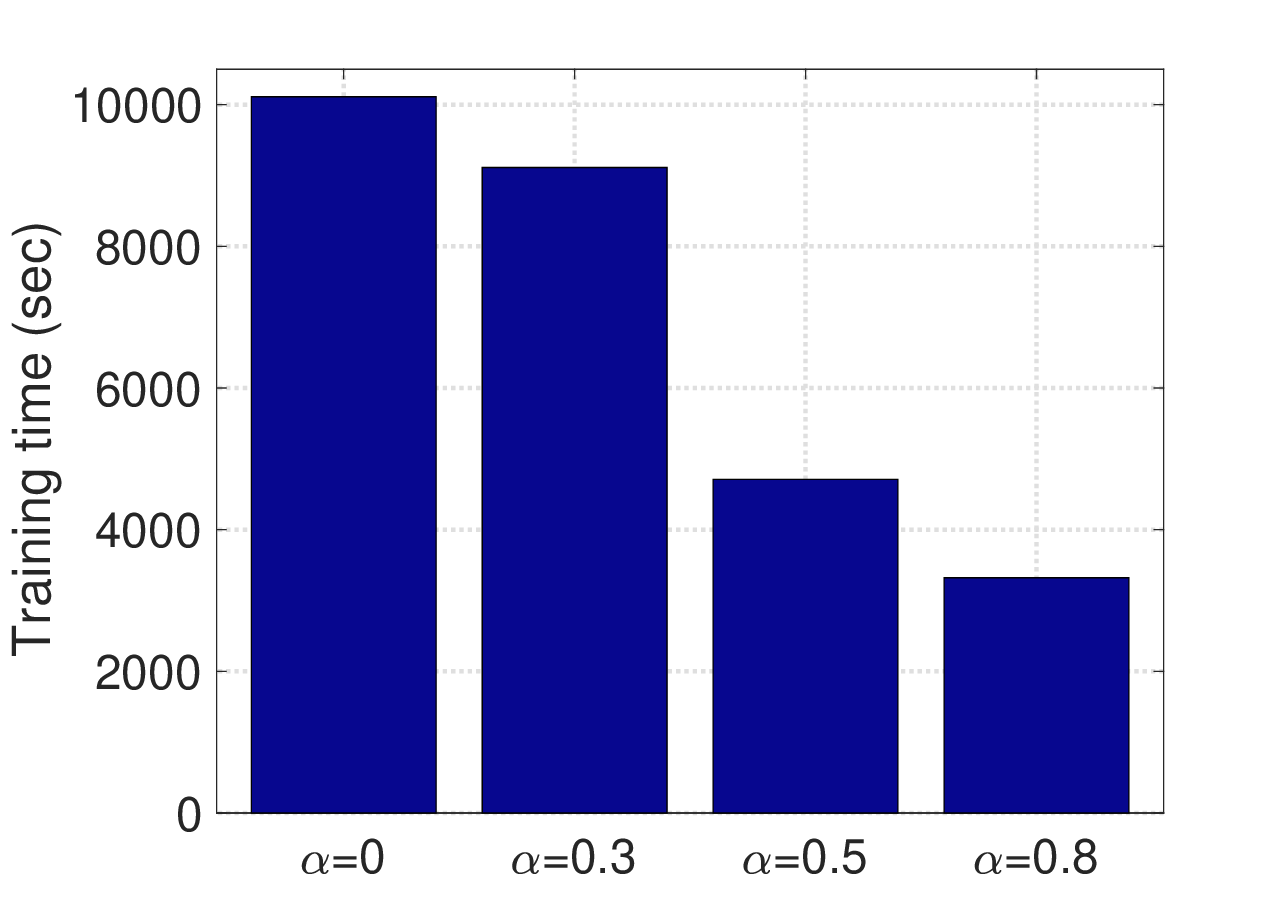}
    \caption{Training time to achieve 88\% accuracy 
 on FMNIST.} 
  \end{subfigure}
     \caption{Effect of the portion of non-sensitive  samples on our solution ($\alpha=0$ reduces to no data offloading).}
\label{fig:non_sensitive}
      \end{figure}

\textbf{Effect of the portion of non-sensitive data:}  In Fig. \ref{fig:non_sensitive}, we also study how the portion of non-sensitive samples $\alpha$ in each ground device's local dataset, affects the FL performance. If all data samples are privacy-sensitive (i.e., $\alpha=0$), the setting reduces to conventional FL with no data offloading. Accuracy curves and the training time required to achieve the target accuracy are reported under the non-IID setting.   We see that our methodology achieves the target accuracy faster as  $\alpha$ increases, since a larger $\alpha$ provides a more flexible data offloading solution for our scheme.

\textbf{Experiments with free-space path loss model:}   In practice, there often exists a line-of-sight link between the ground device and the air node. To validate the effectiveness of our approach under this setting, we use the free-space path loss model between the ground device and the air node, as adopted in \cite{fu2023federated, wu2018joint}, considering that the line-of-sight link is dominant. We also adopt this free-space path loss model for satellite communication, where there is always a line-of-sight link. Fig. \ref{fig:LOS} shows the
results using the CIFAR-10 dataset in both IID and non-IID scenarios. Compared to the setting with Rayleigh fading in Fig. \ref{fig:exp_overall}, all schemes in Fig. \ref{fig:LOS} achieve faster convergence with less training time due to the reduced communication delay. It can be seen that our scheme consistently
outperforms existing baselines by strategically taking advantage of the resources across space-air-ground integrated networks. The overall results further confirm
the effectiveness and applicability of our method.

\begin{figure}[t]
     \centering
       \begin{subfigure}[b]{0.24\textwidth}
         \centering
\includegraphics[width=\textwidth] {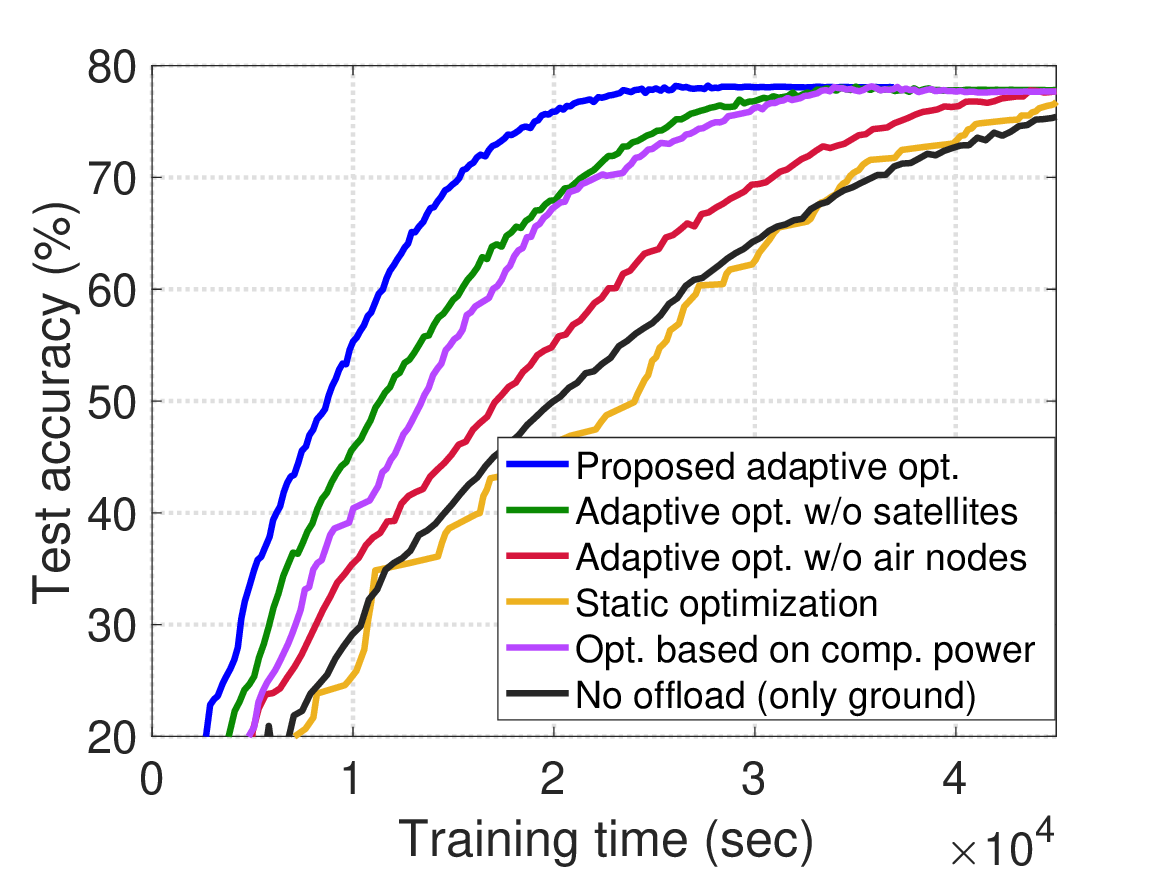}
    \caption{CIFAR-10, IID}
 \end{subfigure}  
          \begin{subfigure}[b]{0.24\textwidth}
         \centering
\includegraphics[width=\textwidth] {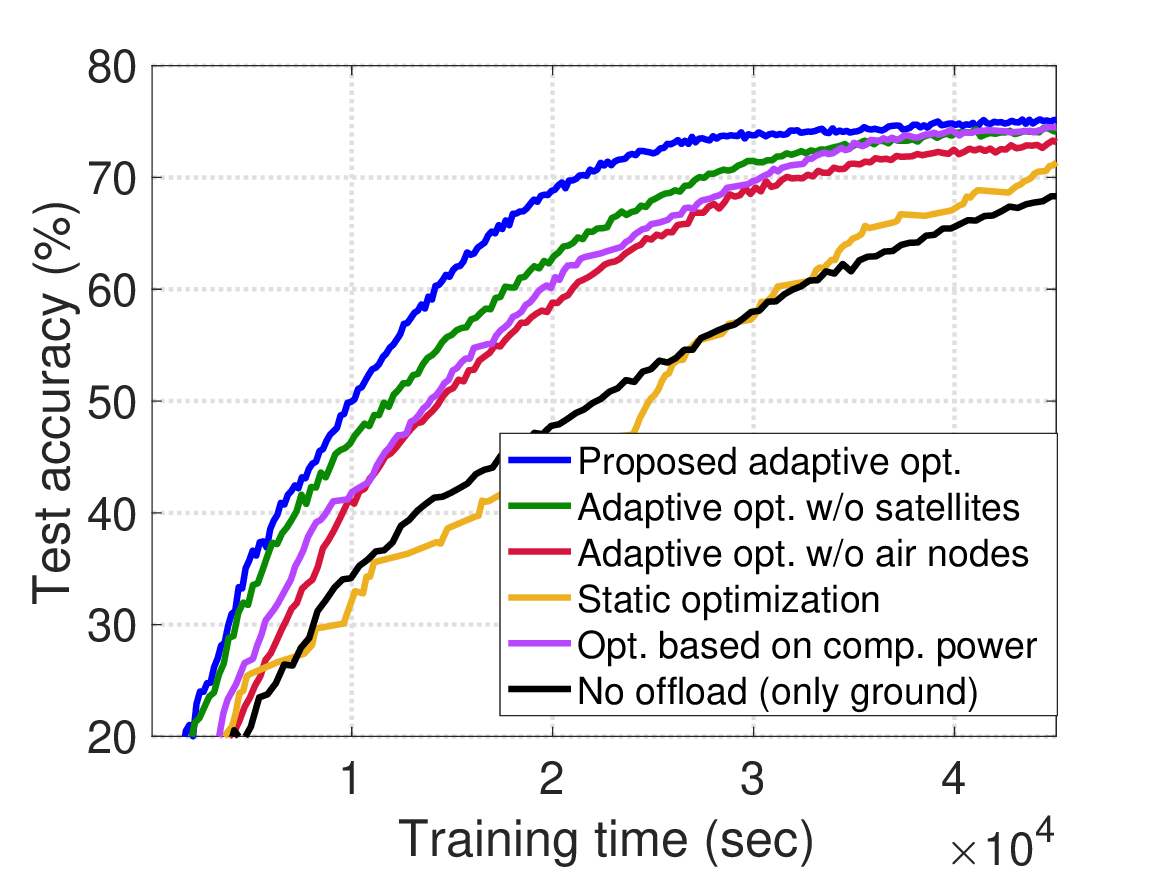}
    \caption{CIFAR-10, Non-IID}
 \end{subfigure} 
     \caption{Experimental results using the free-space pathloss model  with a dominant line-of-sight link.}
    \label{fig:LOS}
    \end{figure}

\section{Conclusion and Future Directions} \label{sec:conclusion}
In this paper, we proposed a distributed ML methodology that orchestrates FL in space-air-ground integrated networks. The core idea was to take advantage of both computation and communication resources of different layers in SAGINs to facilitate/accelerate FL in remote regions. We analytically characterized the latency of our method, and proposed an  adaptive data offloading solution to minimize the training time depending on the current resource availability. We also derived the convergence bound of the scheme and guaranteed convergence to a stationary point  for non-convex loss functions. The advantages of the proposed method as well as the effects of system parameters are investigated via simulations.

There are several promising  directions for future research in this domain. One  direction is to optimize the trajectories of air nodes to achieve a better performance within our framework. Another direction is to introduce an additional layer by considering   the base stations or   geostationary earth orbit satellites  that can  connect  to the LEO satellites, to  further enhance the performance. 
\appendix
\subsection{Proof of Theorem \ref{theorem_full}}\label{proof_theorem}
For ease of notation, we adopt the following equivalent form for the global loss function: 
\equa{\label{global_FL}
	   F(\mathbf{w})  & \triangleq \sum_{i\in \mathcal{P}} \lambda_{i}^{(r)} \ell_{i}^{(r+1)}(\mathbf{w}),
}
where $\mathcal{P} := \left\{(\textsf{G}, k) \mid k\in \mathcal{G} \right\} \cup \left\{(\textsf{A}, n) \mid n  \in \mathcal{A} \right\}\cup\{\textsf{S}\}$ and $\sum_{i\in \mathcal{P}} \lambda_{i}^{(r)} = 1$.
Additionally, we define \begin{small}$\Phi_r = \sum_{h=0}^{H-1} \sum_{i\in \mathcal{P}} \lambda_{i}^{(r)} \mathbb{E} \left\| \mathbf{w}^{(r,h)}_{i} - \mathbf{w}^{(r)} \right\|^2,$\end{small} where $\mathbf{w}^{(r,h)}_{i}$ represents the models parameter of device $i$ after $h$ local iterations within the $r$-th round. It denotes the intermediate model in (\ref{eq:local__GG}), (\ref{eq:local__AA}), or (\ref{eq:update_sat}).
To prove the convergence of the proposed algorithm, we first investigate how each round of training reduces the global loss, as formalized in Lemma \ref{progress_lem_full}. 
\begin{lemma}\label{progress_lem_full}
Under Assumptions \ref{assump_smooth}-\ref{bound_heterogeneity} and  $\eta^{(r)} \leq \frac{1}{2HL}$, we have
\vspace{-2mm}
\begin{small}
\begin{align}\label{progress_full_equation}
\mathbb{E} \!\left[ \!F\left(\mathbf{w}^{(r\!+\!1)}\right) \!\right] 
\! \leq & \mathbb{E}\left[F\left(\mathbf{w}^{(r)}\right)\right] 
\!-\! \frac{\eta^{(r)} H}{2}  \mathbb{E}\left\|\nabla F\left(\mathbf{w}^{(r)}\right) \right\|^2 \nonumber \\
& +\!  \frac{\eta^{(r)} L^2}{2} \!\Phi_r \!+\! (\eta^{(r)})^2 HL \sigma_g^2 \sum_{i\in \mathcal{P}} (\lambda_{i}^{(r)})^2. 
\end{align}
\end{small}
\end{lemma}

To characterize the evolution of $\mathbb{E}\left\|\nabla F\left(\mathbf{w}^{(r)}\right) \right\|^2$ as shown in Theorem \ref{theorem_full}, we need  to further   bound the term $\Phi_r = \sum_{h=0}^{H-1} \sum_{i\in \mathcal{P}} \lambda_{i}^{(r)} \mathbb{E} \left\| \mathbf{w}^{(r,h)}_{i} \!-\! \mathbf{w}^{(r)} \right\|^2$ that appears in Lemma \ref{progress_lem_full}. We establish an upper bound for $\Phi_r$ in Lemma \ref{drift_lem}.
\begin{lemma}
\label{drift_lem}
Under Assumptions \ref{assump_smooth}-\ref{bound_heterogeneity} and $\eta^{(r)} \leq \frac{1}{2HL}$, we have 
\begin{small}
\equa{\label{}
 \Phi_r  \!\leq \!2(1\!+\!c_r)H^3(\eta^{(r)})^2 \mathbb{E} \!\left \|\nabla F(\mathbf{w}^{(r)}) \right \|^2 \!+\! \frac{2}{3}H^3(\eta^{(r)})^2 (\sigma_g^2 \!+\! 3 \delta_r^2). \nonumber 
}
\end{small}
\end{lemma}

The proofs of Lemmas \ref{progress_lem_full} and \ref{drift_lem} are provided in Appendix \ref{proof_lemmas}. Combining  Lemmas \ref{progress_lem_full} and \ref{drift_lem}, we obtain 
\vspace{-4mm}

\begin{small}
\begin{align}\label{}
& \mathbb{E}\! \left[ \! F\left(\mathbf{w}^{(r\!+\!1)}\right) \! \right] 
  \leq \mathbb{E}\! \left[ \! F\left(\mathbf{w}^{(r)}\right) \! \right]  
\!-\! \frac{\eta^{(r)} H}{2}  \mathbb{E} \left\|\nabla F\left(\mathbf{w}^{(r)}\right)\right\|^2  \nonumber \\
& +\!  (\eta^{(r)})^2 HL \sigma_g^2 \sum_{i\in \mathcal{P}} (\lambda_{i}^{(r)})^2 \!+\!  \frac{\eta^{(r)} L^2}{2} \left\{2(1+c_r)H^3(\eta^{(r)})^2 \right. \nonumber \\
&  \left. \times \mathbb{E} \left \|\nabla F(\mathbf{w}^{(r)}) \right \|^2 + \frac{2}{3}H^3(\eta^{(r)})^2 (\sigma_g^2 + 3 \delta_r^2)\right\}. \nonumber 
\end{align}
\end{small}

Reorganizing the above inequality and utilizing (\ref{theorem_eta_full}) give rise to the following result:

\vspace{-4mm}

\begin{small}
\begin{align}\label{}
& \eta^{(r)} \mathbb{E} \left\|\nabla F\left(\mathbf{w}^{(r)}\right)\right\|^2 \leq 4 \frac{\mathbb{E} \! \left[ \!F\left(\mathbf{w}^{(r)}\right)\!\right] \!-\!  \mathbb{E} \left[ \!F\left(\mathbf{w}^{(r\!+\!1)}\right) \! \right]}{ H } \nonumber \\
&+\! 4 (\eta^{(r)})^2 L \sigma_g^2 \sum_{i\in \mathcal{P}} (\lambda_{i}^{(r)})^2 \!+\!  2(\eta^{(r)})^3 H^2 L^2 \sigma_g^2 + 4(\eta^{(r)})^3 H^2 L^2 \delta_r^2. \nonumber 
\end{align}
\end{small}

By telescopic expansion of the above inequality from $r=0$ to $R-1$, we can obtain the result shown in Theorem \ref{theorem_full}.

\vspace{-2mm}
\subsection{Proof of Lemmas}\label{proof_lemmas}

\subsubsection{Proof of Lemma \ref{progress_lem_full}}
 For ease of notation, we denote $\bm e_i^{(r,h)}, \ i \in \mathcal{P} = \left\{(\textsf{G}, k) \mid k\in \mathcal{G} \right\} \cup \left\{(\textsf{A}, n) \mid n  \in \mathcal{A} \right\}\cup\{\textsf{S}\}$ as a mini-batch gradient $\tilde{\nabla}\ell_{\textsf{G}, k}^{(r+1)}(\mathbf{w}^{(r,h)}_{\textsf{G}, k}), k\in \mathcal{G}$, $\tilde{\nabla}\ell_{\textsf{A}, n}^{(r+1)}(\mathbf{w}^{(r,h)}_{\textsf{A}, n}), n\in \mathcal{A}$, or $\tilde{\nabla}\ell_{\textsf{S}}^{(r+1)}(\mathbf{w}^{(r,h)}_{\textsf{S}})$. 

Due to the smoothness of local loss functions described in Assumption \ref{assump_smooth}, the global loss function $F(\mathbf{w})$ is $L$-smooth as well.
Based on the iteration 
$\sum_{h=0}^{H-1}  \sum_{i \in \mathcal{P}} \lambda_{i}^{(r)} \bm e_i^{(r,h)} = \mathbf{w}^{(r\!+\!1)} - \mathbf{w}^{(r)}$, we have
\vspace{-4mm}

\begin{small}
\begin{align}\label{descent_lemma_1}
& \mathbb{E} \!\left[ \!F\left(\mathbf{w}^{(r\!+\!1)}\right) \!\right] 
 \!\leq \! \mathbb{E} \!\left[\!F\left(\mathbf{w}^{(r)}\right)\!\right] \!+\! (\eta^{(r)})^2\frac{L}{2} \underbrace{\mathbb{E} \! \left \|\! \sum_{h=0}^{H-1} \! \sum_{i \in \mathcal{P}} \! \lambda_{i}^{(r)} \bm e_i^{(r,h)} \! \right\|^2 }_{\Psi_2}
 \nonumber \\
 & \qquad \underbrace{- \eta^{(r)}  \mathbb{E} \left \langle \! \nabla \!F\left(\mathbf{w}^{(r)}\right), 
 \sum_{h=0}^{H-1} \! \sum_{i \in \mathcal{P}} \! \lambda_{i}^{(r)} \bm e_i^{(r,h)} \! \right  \rangle }_{\Psi_1}.  
\end{align}
\end{small}

We next bound $\Psi_1$ and $\Psi_2$. 
First, for $\Psi_1$, we have 
\begin{small}
	\[
\Psi_1 \!=\! - \eta^{(r)} H \mathbb{E} \left \langle \nabla F\left(\mathbf{w}^{(r)}\right) , \frac{1}{H} \sum_{h=0}^{H-1}  \sum_{i \in \mathcal{P}} \lambda_{i}^{(r)} {\nabla} \ell_{i}^{(r\!+\!1)}(\mathbf{w}^{(r,h)}_{i})  \right \rangle.
\]
\end{small}

Due to \begin{small}
    $-\langle \mathbf{a},  \mathbf{b}\rangle = \frac{1}{2} \|\mathbf{a}-\mathbf{b}\|^2 - \frac{1}{2}\|\mathbf{a}\|^2 - \frac{1}{2}\|\mathbf{b}\|^2$
\end{small}, we have 
\begin{small}
\equa{\label{inne_gra_app}
& \Psi_1 \!=\! -\! \frac{\eta^{(r)}\! H}{2} \!\left\{ \!\mathbb{E} \!\left\|\!\frac{1}{H} \!\sum_{h=0}^{H\!-\!1}\! \sum_{i \in \mathcal{P}} \! \lambda_{i}^{(r)} {\nabla} \ell_{i}^{(r\!+\!1)}\!(\mathbf{w}^{(r,h)}_{i}) \!\right\|^2 \!+\! \mathbb{E} \left\|\!\nabla \!F\left(\!\mathbf{w}^{(r)}\!\right) \!\right\|^2 \!\right\} \\
&\! +\! \frac{\eta^{(r)} H}{2} \underbrace{ \mathbb{E} \left\|\frac{1}{H} \sum_{h=0}^{H-1} \sum_{i \in \mathcal{P}} \lambda_{i}^{(r)} {\nabla} \ell_{i}^{(r\!+\!1)}(\mathbf{w}^{(r,h)}_{i})- \nabla F(\mathbf{w}^{(r)}) \right\|^2 }_{\Psi_3}. \nonumber
}
\end{small}

Now based on   $ \sum_{i \in \mathcal{P}} \lambda_{i}^{(r)} {\nabla} \ell_{i}^{(r\!+\!1)}(\mathbf{w}^{(r)}) = \nabla F(\mathbf{w}^{(r)})$, $\sum_{i \in \mathcal{P}} \lambda_{i}^{(r)}  = 1$,  the Jensen's inequality, and Assumption \ref{assump_smooth}, we can  bound $\Psi_3$ as
\vspace{-4mm}

\begin{small}
\begin{align}
\Psi_3 
\leq & \frac{L^2}{H} \sum_{h=0}^{H-1} \sum_{i\in \mathcal{P}} \lambda_{i}^{(r)} \mathbb{E} \left\| \mathbf{w}^{(r,h)}_{i} - \mathbf{w}^{(r)} \right\|^2, \nonumber 
\end{align}
\end{small}
where the inequality comes from Assumption \ref{assump_smooth}.

For $\Psi_2$, by using the Cauchy-Schwartz inequality, we have
\vspace{-4mm}

\begin{small}
\begin{align}
\Psi_2 \leq & 2 \mathbb{E} \left \| \sum_{h=0}^{H-1} \sum_{i\in \mathcal{P}} \lambda_{i}^{(r)}  \left(\tilde{\nabla} \ell_{i}^{(r\!+\!1)}(\mathbf{w}^{(r,h)}_{i}) - {\nabla} \ell_{i}^{(r\!+\!1)}(\mathbf{w}^{(r,h)}_{i}) \right) \right\|^2 \nonumber \\
& + 2 \mathbb{E} \left \|\sum_{h=0}^{H-1} \sum_{i\in \mathcal{P}} \lambda_{i}^{(r)} {\nabla} \ell_{i}^{(r\!+\!1)}(\mathbf{w}^{(r,h)}_{i}) \right\|^2 \nonumber \\
\leq & 2 H \sigma_g^2 \sum_{i\in \mathcal{P}} (\lambda_{i}^{(r)})^2 + 2 \mathbb{E} \left \|\sum_{h=0}^{H-1} \sum_{i\in \mathcal{P}} \lambda_{i}^{(r)} {\nabla} \ell_{i}^{(r\!+\!1)}(\mathbf{w}^{(r,h)}_{i}) \right\|^2.\nonumber
\end{align}
\end{small}

 Utilizing $\eta^{(r)} \leq \frac{1}{2HL}$ and combining $\Psi_1$, $\Psi_2$, and $\Psi_3$ with (\ref{descent_lemma_1})  give rise to Lemma \ref{progress_lem_full}.

\subsubsection{Proof of Lemma \ref{drift_lem}}
We first denote $s^{(r,\tau)} = \sum_{i\in \mathcal{P}} \lambda_{i}^{(r)} \mathbb{E} \left\| \mathbf{w}^{(r,\tau)}_{i} - \mathbf{w}^{(r)} \right\|^2$, which can be bounded as

\begin{small}
\begin{align}
&s^{(r,\tau)} \!\!=\! (\eta^{(r)})^2 \!\sum_{i\in \mathcal{P}} \!\lambda_{i}^{(r)}  \mathbb{E} \!\left\|  \sum_{k=0}^{\tau-1} \!\bm e_i^{(t,k)} \!\right\|^2 
\!\!\leq \!\! \tau (\eta^{(r)})^2 \!\sum_{h=0}^{\tau-1} \!\sum_{i\in \mathcal{P}} \!\lambda_{i}^{(r)}\!
\mathbb{E}\! \left\|  \bm e_i^{(r,h)} \!\right\|^2 \nonumber \\
=& \tau (\eta^{(r)})^2 \sum_{h=0}^{\tau-1} \sum_{i\in \mathcal{P}} \lambda_{i}^{(r)} \mathbb{E} \left \|
 \bm e_i^{(r,h)} - \nabla {\ell}_i^{(r\!+\!1)}(\mathbf{w}_i^{(r,h)}) \right\|^2 \nonumber \\ 
 &+ \tau (\eta^{(r)})^2 \sum_{h=0}^{\tau-1} \sum_{i\in \mathcal{P}} \lambda_{i}^{(r)} \mathbb{E} \left \| \nabla {\ell}_i^{(r\!+\!1)}(\mathbf{w}_i^{(r,h)}) \right\|^2 \nonumber \\
 \leq & \tau (\eta^{(r)})^2 \sum_{h=0}^{\tau-1} \underbrace{ \sum_{i\in \mathcal{P}} \lambda_{i}^{(r)} \mathbb{E} \left \| \nabla {\ell}_i^{(r\!+\!1)}(\mathbf{w}_i^{(r,h)})  \right\|^2}_{\Psi_4} +  (\eta^{(r)})^2 \tau^2  \sigma_g^2. \label{T_4_T_6}
\end{align}
\end{small}

Next, we establish an upper bound for $\Psi_4$ as
\vspace{-4mm}

\begin{small}
\begin{align}
\Psi_4
&= \sum_{i\in \mathcal{P}} \lambda_{i}^{(r)} \mathbb{E} \left \| \nabla {\ell}_i^{(r\!+\!1)}(\mathbf{w}_i^{(r,h)}) \!\mp\! \nabla {\ell}_i^{(r\!+\!1)}(\mathbf{w}^{(r)})   \! \mp\! \nabla F(\mathbf{w}^{(r)}) \right\|^2 \nonumber\\
\leq & 3L^2 \!\sum_{i\in \mathcal{P}} \lambda_{i}^{(r)} \mathbb{E} \!\left \|\! \mathbf{w}_i^{(r,h)} \!-\! \mathbf{w}^{(r)} \!\right\|^2  \!+\! (3\!+\!3c_r) \mathbb{E} \!\left \|\!F(\mathbf{w}^{(r)}) \!\right\|^2\!+\!  3\delta_r^2, \nonumber
\end{align}
\end{small}

\noindent where the last inequality comes from Assumption \ref{bound_heterogeneity}. By plugging the upper bound of $\Psi_4$  into (\ref{T_4_T_6}) and taking summation over $\tau$ from $1$ to $H-1$, we obtain

\vspace{-4mm}

\begin{small}
\begin{align}
\sum_{\tau=1}^{H-1} \!s^{(r,\tau)} 
\!
\leq & 2 H^2 L^2 (\eta^{(r)})^2  \sum_{h=0}^{H\!-\!1}\! s^{(r,h)} 
\!+\! (1\!+\!c_r) \!H^3\!(\eta^{(r)})^2  \nonumber \\
&\times \mathbb{E} \left \|\nabla F(\mathbf{w}^{(r)}) \right \|^2 +  H^3(\eta^{(r)})^2 (\frac{1}{3}\sigma_g^2 + \delta_r^2),\label{iterate_s}
\end{align}
\end{small}
\vspace{-4mm}

\noindent where we utilize the property of arithmetic sequence.
Utilizing $s^{(r,0)}\!=\!0$ and rearranging (\ref{iterate_s}), we have
\begin{small}
\equa{
(1\!-\!2H^2 L^2 (\eta^{(r)})^2)  \sum_{\tau=0}^{H-1} s^{(r,\tau)} 
\leq & (1+c_r)H^3(\eta^{(r)})^2 \mathbb{E} \left \| \nabla F(\mathbf{w}^{(r)}) \right \|^2 \\
&+ H^3(\eta^{(r)})^2 (\frac{1}{3}\sigma_g^2 + \delta_r^2). \nonumber
}
\end{small}

Since $\eta^{(r)} \leq \frac{1}{2HL}$ holds, we have $(1-2 H^2 L^2 (\eta^{(r)})^2) \geq \frac{1}{2}$. 
Scaling the above inequality gives rise to Lemma \ref{drift_lem}.

\bibliography{Reference_space_air_ground}
\bibliographystyle{IEEEtran}
\begin{IEEEbiography}[{\includegraphics[width=1in,height=1.25in,clip,keepaspectratio]{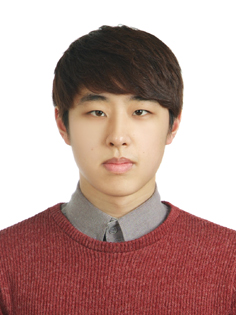}}]{Dong-Jun Han}  (Member, IEEE) is an Assistant Professor of School of Computer Science and Engineering at Yonsei University, South Korea. Previously, he was a  postdoctoral researcher  in the School of Electrical and Computer Engineering at  Purdue University. He received the B.S. degrees in mathematics and electrical engineering, and the M.S. and Ph.D. degrees  in electrical engineering from Korea Advanced Institute of Science and Technology (KAIST), South Korea, in 2016,  2018, and 2022, respectively.   His research interest is at the intersection of communications, networking, and machine learning, specifically in distributed/federated machine learning and network optimization. 
\end{IEEEbiography}

\begin{IEEEbiography}[{\includegraphics[width=1in,height=1.25in,clip,keepaspectratio]{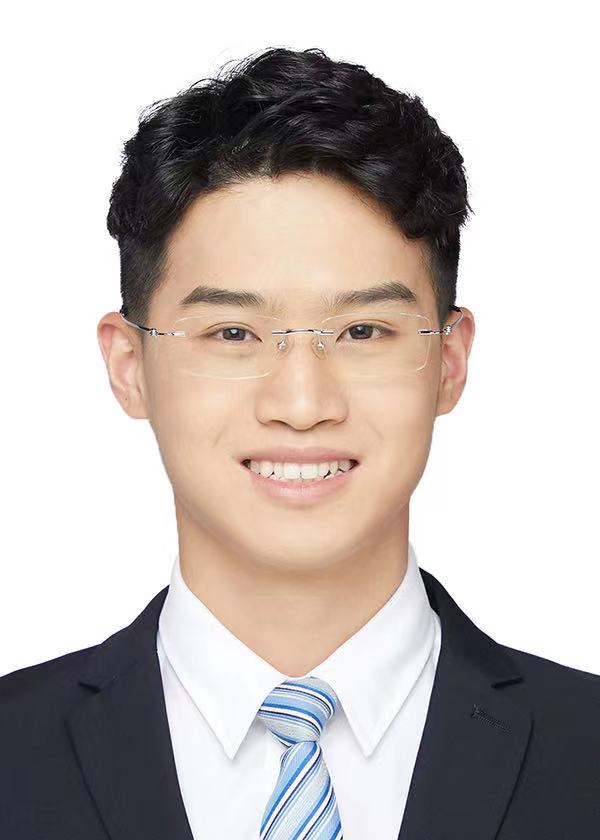}}]{Wenzhi Fang}  (Graduate Student Member, IEEE) 
 received his B.S. degree from Shanghai University in 2020 and completed his master’s degree at ShanghaiTech University in 2023. Currently, he is pursuing a PhD in electrical and computer engineering at Purdue University, West Lafayette, US. His research interests focus on optimization theory and its applications in machine learning, signal processing, and wireless networks.
\end{IEEEbiography}

\begin{IEEEbiography}
[{\includegraphics[width=1in,height=1.25in,clip,keepaspectratio]{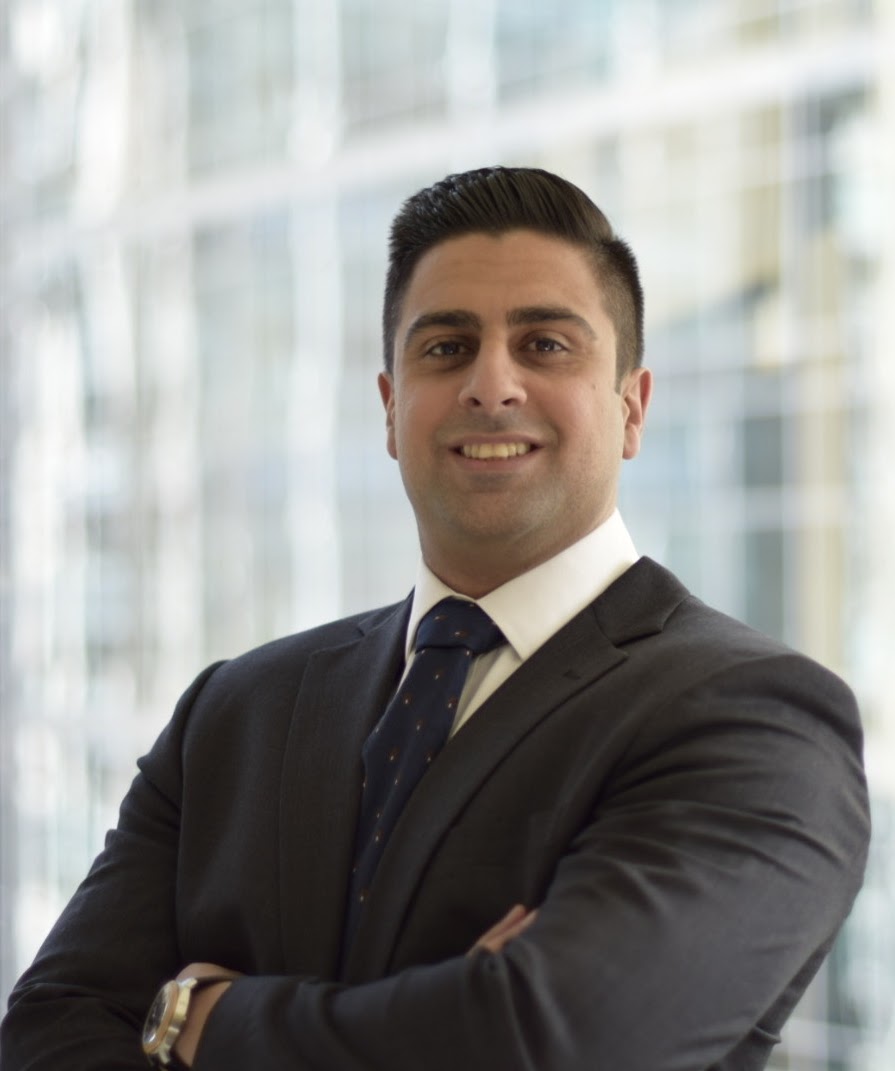}}]{Seyyedali Hosseinalipour} (Member, IEEE)  received the B.S. degree in electrical engineering
from Amirkabir University of Technology, Tehran, Iran, in 2015 with high honor
and top-rank recognition. He then received the M.S. and Ph.D. degrees in electrical engineering from North Carolina State University, NC, USA, in 2017 and
2020, respectively. He was the recipient of the ECE Doctoral Scholar of the Year
Award (2020) and ECE Distinguished Dissertation Award (2021) at North Carolina State University. He was a postdoctoral researcher at Purdue University,
IN, USA from 2020 to 2022. He is currently an assistant professor at the Department of Electrical Engineering at the University at Buffalo (SUNY). He has
served as the TPC Co-Chair of workshops and symposiums related to distributed machine learning and edge computing held in conjunction with IEEE INFOCOM, IEEE GLOBECOM, IEEE ICC, IEEE/CVF CVPR, IEEE MSN, and IEEE VTC. Also, he has served as the guest editor for IEEE Internet of Things
Magazine. His research interests include the analysis of modern wireless networks,
synergies between machine learning methods and fog computing systems, distributed/federated machine learning, and network optimization.
\end{IEEEbiography}

\begin{IEEEbiography}
[{\includegraphics[width=1.0in,height=1.28in,clip,keepaspectratio]{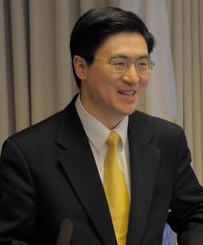}}]{Mung Chiang} (Fellow, IEEE) is the 13th President of Purdue University and the Roscoe H. George Distinguished Professor of Electrical and Computer Engineering. Previously he was the Arthur LeGrand Doty Professor of Electrical Engineering at Princeton University, where he founded the Princeton Edge Lab in 2009 and cofounded several startups spun out from there. The 2013 NSF Alan T. Waterman Awardee, he also received a Guggenheim Fellowship, the IEEE Kiyo Tomiyau Award, the IEEE INFOCOM Achievement Award, and is a member of the National Academy of Inventors and Royal Swedish Academy of Engineering Science. He served as the Science and Technology Adviser to the U.S. Secretary of State. 
\end{IEEEbiography}

\begin{IEEEbiography}
[{\includegraphics[width=1in,height=1.25in,clip,keepaspectratio]{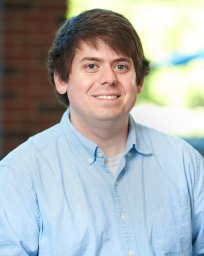}}]{Christopher G. Brinton} (Senior Member, IEEE) is the Elmore Rising Star Associate Professor of Electrical and Computer Engineering (ECE) at Purdue University. His research interest is at the intersection of networking, communications, and machine learning, specifically in fog/edge network intelligence, distributed machine learning, and AI/ML-inspired wireless network optimization. Dr. Brinton is a recipient of five of the US top early career awards, from the National Science Foundation (CAREER), Office of Naval Research (YIP), Defense Advanced Research Projects Agency (YFA and Director’s Fellowship), and Air Force Office of Scientific Research (YIP), the IEEE Communication Society William Bennett Prize Best Paper Award, the Intel Rising Star Faculty Award, the Qualcomm Faculty Award, and roughly \$17M in sponsored research projects as a PI or co-PI. He has also been awarded Purdue College of Engineering Faculty Excellence Awards in Early Career Research, Early Career Teaching, and Online Learning. He currently serves as an Associate Editor for IEEE/ACM Transactions on Networking, and previously was an Associate Editor for IEEE Transactions on Wireless Communications. Prior to joining Purdue, Dr. Brinton was the Associate Director of the EDGE Lab and a Lecturer of Electrical Engineering at Princeton University. He also co-founded Zoomi Inc., a big data startup company that holds US Patents in machine learning for education. His book The Power of Networks: 6 Principles That Connect our Lives and associated Massive Open Online Courses (MOOCs) reached over 400,000 students. Dr. Brinton received the PhD (with honors) and MS Degrees from Princeton in 2016 and 2013, respectively, both in Electrical Engineering. 
\end{IEEEbiography}
\end{document}